\documentclass[11pt]{article}

\usepackage[letterpaper,margin=1.25in]{geometry}
\usepackage[T1]{fontenc}
\usepackage{microtype}
\usepackage[style=authoryear,giveninits=true,uniquename=false,uniquelist=false,maxbibnames=99]{biblatex}
\usepackage[title]{appendix}
\usepackage[colorlinks=true,citecolor=blue,urlcolor=blue]{hyperref}
\usepackage{mathtools}
\usepackage{amsthm}
\usepackage{amssymb}
\usepackage{bm}
\usepackage{tikz}
\usepackage[labelfont=bf]{caption}
\usepackage{booktabs}
\usepackage[ruled]{algorithm2e}
\usepackage{authblk}

\renewbibmacro{in:}{}
\DeclareNameAlias{sortname}{family-given}
\DeclareNameAlias{default}{family-given}

\newcommand*{\email}[1]{\href{mailto:#1}{\nolinkurl{#1}}}
\newtheorem{theorem}{Theorem}
\newtheorem{proposition}{Proposition}
\newtheorem{lemma}{Lemma}
\newtheorem{assumption}{Assumption}
\newtheorem{definition}{Definition}

\DeclareMathOperator{\p}{P}

\DeclareMathOperator{\argmin}{arg\,min}
\DeclareMathOperator{\gs}{gs}

\numberwithin{equation}{section}

\newcommand\keywords[1]{%
\begin{NoHyper}
\renewcommand\thefootnote{}\footnote{\emph{Keywords:} #1}%
\addtocounter{footnote}{-1}%
\end{NoHyper}
}

\title{Group selection and shrinkage: \\ Structured sparsity for semiparametric additive models}
\author{Ryan Thompson\thanks{Corresponding author. Now at School of Mathematics and Statistics, University of New South Wales and Data61, Commonwealth Scientific and Industrial Research Organisation. Email: \email{ryan.thompson1@unsw.edu.au}} }
\author{Farshid Vahid}
\affil{Department of Econometrics and Business Statistics, Monash University}

\begin{document}

\maketitle

\begin{abstract}
Sparse regression and classification estimators that respect group structures have application to an assortment of statistical and machine learning problems, from multitask learning to sparse additive modeling to hierarchical selection. This work introduces structured sparse estimators that combine group subset selection with shrinkage. To accommodate sophisticated structures, our estimators allow for arbitrary overlap between groups. We develop an optimization framework for fitting the nonconvex regularization surface and present finite-sample error bounds for estimation of the regression function. As an application requiring structure, we study sparse semiparametric additive modeling, a procedure that allows the effect of each predictor to be zero, linear, or nonlinear. For this task, the new estimators improve across several metrics on synthetic data compared to alternatives. Finally, we demonstrate their efficacy in modeling supermarket foot traffic and economic recessions using many predictors. These demonstrations suggest sparse semiparametric additive models, fit using the new estimators, are an excellent compromise between fully linear and fully nonparametric alternatives. All of our algorithms are made available in the scalable implementation \texttt{grpsel}.
\end{abstract}


\section{Introduction}
\label{sec:intro}

Sparsity over group structures arises in connection with a myriad of statistical and machine learning problems, e.g., multitask learning \parencite{Obozinski2006}, sparse additive modeling \parencite{Ravikumar2009}, and hierarchical selection \parencite{Lim2015}. Even sparse linear modeling can involve structured sparsity, such as when a categorical predictor is represented as a sequence of dummy variables. In certain domains, groups may emerge naturally, e.g., disaggregates of the same macroeconomic series or genes of the same biological path. The prevalence of such problems motivates principled estimation procedures capable of encoding structure into the fitted models they produce.

Given response $\mathbf{y}=(y_1,\ldots,y_n)^\top\in\mathbb{R}^n$, predictors $\mathbf{X}=(\mathbf{x}_1,\ldots,\mathbf{x}_n)^\top\in\mathbb{R}^{n\times p}$, and nonoverlapping groups $\mathcal{G}_1,\ldots,\mathcal{G}_g\subseteq\{1,\ldots,p\}$, group lasso \parencite{Yuan2006,Meier2008} solves
\begin{equation*}
\underset{\bm{\beta}\in\mathbb{R}^p}{\min}\,\sum_i\ell\left(\mathbf{x}_i^\top\bm{\beta},y_i\right)+\sum_k\lambda_k\|\bm{\beta}_k\|,
\end{equation*}
where $\ell:\mathbb{R}^2\to\mathbb{R}_+$ is a loss function (e.g., square loss for regression or logistic loss for classification), $\lambda_1,\ldots,\lambda_g$ are nonnegative tuning parameters, and $\bm{\beta}_k\in\mathbb{R}^{p_k}$ are the coefficients $\bm{\beta}$ indexed by $\mathcal{G}_k$.\footnote{Here and throughout, the intercept term is omitted to facilitate exposition.} Group lasso couples coefficients via their $l_2$-norm so that all predictors in a group are selected together.

Just as lasso \parencite{Tibshirani1996} is the continuous relaxation of the combinatorially-hard problem of best subset selection (``best subset''), so is group lasso the relaxation of the combinatorial problem of group subset selection (``group subset''):
\begin{equation*}
\underset{\bm{\beta}\in\mathbb{R}^p}{\min}\,\sum_i\ell\left(\mathbf{x}_i^\top\bm{\beta},y_i\right)+\sum_k\lambda_k1(\|\bm{\beta}_k\|\neq0).
\end{equation*}
Unlike group lasso, which promotes group sparsity implicitly by nondifferentiability of the $l_2$-norm at the null vector, group subset explicitly penalizes the number of nonzero groups. Consequently, one might interpret group lasso as a compromise made in the interest of computation. However, group lasso has a trick up its sleeve that group subset does not: shrinkage. Shrinkage estimators such as lasso are more robust than best subset to noise \parencite{Breiman1996,Hastie2020}. This consideration motivates one to shrink the group subset estimator:
\begin{equation}
\label{eq:shrinkgroupsubset}
\underset{\bm{\beta}\in\mathbb{R}^p}{\min}\,\sum_i\ell\left(\mathbf{x}_i^\top\bm{\beta},y_i\right)+\sum_k\lambda_{0k}1(\|\bm{\beta}_k\|\neq0)+\sum_k\lambda_{1k}\|\bm{\beta}_k\|.
\end{equation}
In contrast to group lasso and group subset, \eqref{eq:shrinkgroupsubset} directly controls both group sparsity and shrinkage via separate penalties---selection via group subset and shrinkage via group lasso. The combination of best subset and lasso in the unstructured setting results in good predictive models with low false positive selection rates \parencite{Mazumder2023}.

Unfortunately, the estimators~\eqref{eq:shrinkgroupsubset}, including group lasso and group subset as special cases, do not accommodate overlap among groups. Specifically, if two groups overlap, one cannot be selected independently of the other. To encode sophisticated structures, such as hierarchies or graphs, groups must often overlap. To address this issue, one can introduce group-specific vectors $\bar{\bm{\nu}}_k\in\mathbb{R}^p$ $(k=1,\ldots,g)$ that are zero everywhere except at the positions indexed by $\mathcal{G}_k$. Letting $\mathcal{V}$ be the set of all tuples $\bar{\bm{\nu}}:=(\bar{\bm{\nu}}_1,\ldots,\bar{\bm{\nu}}_g)$ with elements satisfying this property, group subset with shrinkage becomes
\begin{equation}
\label{eq:overshrinkgroupsubset}
\underset{\substack{\bm{\beta}\in\mathbb{R}^p,\,\bar{\bm{\nu}}\in\mathcal{V} \\ \bm{\beta}=\sum_{k}\bar{\bm{\nu}}_k}}{\min}\,\sum_i\ell\left(\mathbf{x}_i^\top\bm{\beta},y_i\right)+\sum_k\lambda_{0k}1(\|\bar{\bm{\nu}}_k\|\neq0)+\sum_k\lambda_{1k}\|\bar{\bm{\nu}}_k\|.
\end{equation}
The vectors $\bar{\bm{\nu}}_1,\ldots,\bar{\bm{\nu}}_g$ are a decomposition of $\bm{\beta}$ into a sum of latent coefficients that facilitate selection of overlapping groups. For instance, if three predictors, $x_1$, $x_2$, and $x_3$, are spread across two groups, $\mathcal{G}_1=\{1,2\}$ and $\mathcal{G}_2=\{2,3\}$, then $\beta_1=\bar{\nu}_{1,1}$, $\beta_2=\bar{\nu}_{2,1}+\bar{\nu}_{2,2}$, and $\beta_3=\bar{\nu}_{3,2}$. Since $\beta_2$ has a separate latent coefficient for each group, $\mathcal{G}_1$ or $\mathcal{G}_2$ can be selected independently of the other. This latent coefficient approach originated for group lasso with \textcite{Jacob2009,Obozinski2011}. When all groups are disjoint, \eqref{eq:overshrinkgroupsubset} reduces exactly to \eqref{eq:shrinkgroupsubset}.

This paper develops computational methods and statistical theory for group subset with and without shrinkage. Via the formulation~\eqref{eq:overshrinkgroupsubset}, our work accommodates the general overlapping groups setting. On the computational side, we develop algorithms that scale to compute quality (approximate) solutions of the combinatorial optimization problem. Our framework comprises coordinate descent and local search and applies to general smooth convex loss functions (i.e., regression and classification), building on recent advances for best subset \parencite{Hazimeh2020,Dedieu2021}. In contrast to existing computational methods for group subset \parencite{Guo2014,Bertsimas2016}, which rely on branch-and-bound or commercial mixed-integer optimizers, our methods scale to instances with millions of predictors or groups. We implement our framework in the publicly available \texttt{R} package \texttt{grpsel}. On the statistical side, we establish new error bounds for group subset with and without shrinkage. The bounds apply in the overlapping setting and allow for model misspecification. The analysis sheds light on the advantages of structured sparsity and the benefits of shrinkage.

The new estimators have application to a broad range of statistical problems. We focus on sparse semiparametric additive modeling \parencite{Chouldechova2015,Lou2016}, a procedure that models $y$ as a sum of univariate functions $f_j$:
\begin{equation*}
g(\operatorname{E}[y])=\sum_jf_j(x_j),
\end{equation*}
where $g$ is a given link function (e.g., identity for regression or logit for classification). The appeal of this model is its flexibility. An individual $f_j$ can be a complex nonlinear function, a simple linear function $f_j(x_j)=\beta_jx_j$, or a zero function $f_j(x_j)=0$, the latter implying no effect of $x_j$. \textcite{Chouldechova2015} estimate these models using group lasso with overlapping groups and regression splines. Briefly summarized, this approach takes $f_j(x_j)=\sum_{l=1}^m\beta_{jl}b_l(x_j)$, where $b_l(x_j)$ are the basis functions of an $m$-term regression spline in which $b_1(x_j)=x_j$. Two groups are created for each $x_j$: a linear group containing the linear term $b_1(x_j)$, and a nonlinear group containing all $b_1(x_j),\ldots,b_M(x_j)$. Group lasso is then applied to these overlapping groups. If the $j$th nonlinear group is selected, $f_j$ is fit as a nonlinear function, regardless of whether the $j$th linear group is also selected (if it is, the fit is still nonlinear and the two coefficients on $b_1(x_j)$ are added together). If only the linear group is selected, the fit is linear. If neither group is selected, $x_j$ is excluded from the model.

After conducting synthetic experiments on the efficacy of our estimators in fitting sparse semiparametric additive models (``semiparametric models'') using the above approach, we carry out two empirical studies. The first study involves modeling supermarket foot traffic using sales volumes on different products. Only a fraction of supermarket products are traded in volume, necessitating sparsity. The second study involves modeling recessionary periods in the economy using macroeconomic series. The macroeconomic literature contains many examples of sparse linear modeling \parencite{DeMol2008,Li2014}, yet theory does not dictate linearity. Together these studies suggest sparse semiparametric models are an excellent compromise between fully linear and fully nonparametric alternatives.

Independently and concurrently to this work, \textcite{Hazimeh2023} study computation and theory for group subset with nonoverlapping groups. Their algorithms likewise build on \textcite{Hazimeh2020} but apply only to square loss regression. Also related is \textcite{Zhang2023} who propose a computational ``splicing'' technique for group subset that appears promising, though they do not consider overlapping groups or shrinkage.

\subsection{Organization}

The paper is structured as follows. Section~\ref{sec:algorithms} presents computational methods. Section~\ref{sec:theory} provides statistical theory. Section~\ref{sec:implementation} discusses the \texttt{R} implementation. Section~\ref{sec:simulations} describes simulation experiments. Section~\ref{sec:data} reports data analyses. Section~\ref{sec:conclusion} closes the paper.

\section{Computation}
\label{sec:algorithms}

This section introduces our optimization framework and its key components: coordinate descent and local search. The framework applies to any smooth loss function $\ell(z,y)$ convex in $z$. The discussion below addresses the specific cases of square loss $\ell(z,y)=(y-z)^2/2$, which is suitable for regression, and logistic loss $\ell(z,y)=-y\log(z)-(1-y)\log(1-z)$, which is suitable for classification. Throughout this section, the predictor matrix $\mathbf{X}$ is assumed to have columns with mean zero and unit $l_2$-norm.

\subsection{Problem reformulation}
\label{sec:reformulation}

From a computational perspective, it helps to reformulate the group subset problem~\eqref{eq:overshrinkgroupsubset} as an unconstrained minimization problem involving only the latent coefficients $\bar{\bm{\nu}}$. For this task, we denote by $\bm{\nu}_k\in\mathbb{R}^{p_k}$ the restriction of $\bar{\bm{\nu}}_k$ to the coordinates indexed by group $k$. No information is lost in this restriction since all elements not indexed by $\mathcal{G}_k$ are zero. We also introduce the vector $\bm{\nu}:=(\bm{\nu}_1^\top,\ldots,\bm{\nu}_g^\top)^\top\in\mathbb{R}^{\sum_kp_k}$ formed by vertically concatenating the vectors $\bm{\nu}_1,\ldots,\bm{\nu}_g$. Consider now the unconstrained minimization problem
\begin{equation*}
\underset{\bm{\nu}\in\mathbb{R}^{\sum_kp_k}}{\min}\,F(\bm{\nu}):=L(\bm{\nu})+R(\bm{\nu}).
\end{equation*}
Here, the function $L(\bm{\nu})$ is the loss term:
\begin{equation*}
L(\bm{\nu}):=\sum_i\ell\left(\sum_k\mathbf{x}_{ik}^\top\bm{\nu}_k,y_i\right),
\end{equation*}
where $\mathbf{x}_{ik}$ is the $i$th row of the matrix $\mathbf{X}_k$, with $\mathbf{X}_k$ the restriction of $\mathbf{X}$ to the columns indexed by group $k$. The function $R(\bm{\nu})$ is the regularizer term:
\begin{equation*}
R(\bm{\nu}):=\sum_k\left(\lambda_{0k}1(\|\bm{\nu}_k\|\neq0)+\lambda_{1k}\|\bm{\nu}_k\|\right).
\end{equation*}
Observe that the loss $L(\bm{\nu})$ is exactly equivalent to that in \eqref{eq:overshrinkgroupsubset} since
\begin{equation*}
\sum_i\ell\left(\sum_k\mathbf{x}_{ik}^\top\bm{\nu}_k,y_i\right)=\sum_i\ell\left(\mathbf{x}_i^\top\sum_k\bar{\bm{\nu}}_k,y_i\right)=\sum_i\ell\left(\mathbf{x}_i^\top\bm{\beta},y_i\right).
\end{equation*}
The regularizer $R(\bm{\nu})$ is likewise equivalent because $\|\bm{\nu}_k\|=\|\bar{\bm{\nu}}_k\|$. As the equalities immediately above suggest, it is straightforward to recover $\bm{\beta}$ from $\bm{\nu}$.

We refer to the group support, or active set of groups, of the vector $\bm{\nu}$ as the set of nonzero group indices:
\begin{equation*}
\operatorname{gs}(\bm{\nu}):=\{k\in\{1,\ldots,g\}:\|\bm{\nu}_k\|\neq0\}.
\end{equation*}
The complement of the active set is referred to as the inactive set.

\subsection{Coordinate descent}

Coordinate descent algorithms are optimization routines that minimize along successive coordinate hyperplanes. The coordinate descent scheme developed here iteratively fixes all but one group of coordinates (a coordinate group) and minimizes in the directions of these coordinates.

The objective function $F(\bm{\nu})$ is a sum of smooth convex and discontinuous nonconvex functions and is hence discontinuous nonconvex. The minimization problem with respect to group $k$ is
\begin{equation}
\label{eq:cdopt}
\underset{\bm{\xi}\in\mathbb{R}^{p_k}}{\min}\,F(\bm{\nu}_1,\ldots,\bm{\nu}_{k-1},\bm{\xi},\bm{\nu}_{k+1},\ldots,\bm{\nu}_g).
\end{equation}
The complexity of this coordinate-wise minimization depends on the type of loss function and the properties of the group matrix $\mathbf{X}_k$. In the case of square loss, the minimization involves a least-squares fit in $p_k$ coordinates, taking $O(p_k^2n)$ operations. To bypass these involved computations, a partial minimization scheme is adopted whereby each coordinate group is updated using a single gradient descent step taken with respect to that group. This scheme results from a standard technique of replacing the objective function with a surrogate function that is an upper bound. To this end, we require Lemma~\ref{lemma:descent}.
\begin{lemma}
\label{lemma:descent}
Let $L:\mathbb{R}^{\sum_kp_k}\to\mathbb{R}$ be a continuously differentiable function. Suppose there exists a $c_k>0$ such that the gradient of $L$ with respect to the $k$th coordinate group satisfies the Lipschitz property
\begin{equation*}
\|\nabla_kL(\bm{\nu})-\nabla_kL(\tilde{\bm{\nu}})\|\leq c_k\|\bm{\nu}_k-\tilde{\bm{\nu}}_k\|,
\end{equation*}
for all $\bm{\nu}\in\mathbb{R}^{\sum_kp_k}$ and $\tilde{\bm{\nu}}\in\mathbb{R}^{\sum_kp_k}$ that differ only in group $k$. Then it holds
\begin{equation}
\label{eq:upper}
L(\bm{\nu})\leq\bar{L}_{\bar{c}_k}(\bm{\nu};\tilde{\bm{\nu}}):=L(\tilde{\bm{\nu}})+\nabla_kL(\tilde{\bm{\nu}})^\top(\bm{\nu}_k-\tilde{\bm{\nu}}_k)+\frac{\bar{c}_k}{2}\|\bm{\nu}_k-\tilde{\bm{\nu}}_k\|^2,
\end{equation}
for any $\bar{c}_k\geq c_k$.
\end{lemma}
Lemma~\ref{lemma:descent} is the block descent lemma of \textcite{Beck2013a}, which holds under a Lipschitz condition on the coordinate-wise gradients of $L(\bm{\nu})$. This condition is satisfied for square loss with $c_k=\gamma_k^2$ and for logistic loss with $c_k=\gamma_k^2/4$, where $\gamma_k$ is the maximal eigenvalue of $\mathbf{X}_k^\top\mathbf{X}_k$. Using the result of Lemma~\ref{lemma:descent}, an upper bound of $F(\bm{\nu})$, treated as a function in group $k$, is given by
\begin{equation}
\label{eq:upperreg}
\bar{F}_{\bar{c}_k}(\bm{\nu};\tilde{\bm{\nu}}):=\bar{L}_{\bar{c}_k}(\bm{\nu};\tilde{\bm{\nu}})+R(\bm{\nu}).
\end{equation}
Thus, in place of the minimization~\eqref{eq:cdopt}, we use the minimization
\begin{equation}
\label{eq:cdoptupper}
\underset{\bm{\xi}\in\mathbb{R}^{p_k}}{\min}\,\bar{F}_{\bar{c}_k}(\bm{\nu}_1,\ldots,\bm{\nu}_{k-1},\bm{\xi},\bm{\nu}_{k+1},\ldots,\bm{\nu}_g;\tilde{\bm{\nu}}).
\end{equation}
This new problem admits a simple analytical solution, given by Proposition~\ref{prop:cdupdate}.
\begin{proposition}
\label{prop:cdupdate}
Define the thresholding function
\begin{equation}
\label{eq:threshold}
T_c(\bm{\xi};\lambda_0,\lambda_1):=
\begin{dcases*}
\left(1-\frac{\lambda_1}{c\|\bm{\xi}\|}\right)_+\bm{\xi} & if $\left(1-\frac{\lambda_1}{c\|\bm{\xi}\|}\right)_+\|\bm{\xi}\|\geq\sqrt{\dfrac{2\lambda_0}{c}}$ \\
\bm{0} & otherwise,
\end{dcases*}
\end{equation}
where $(x)_+$ is shorthand for $\max(x,0)$. Then the coordinate-wise minimization problem~\eqref{eq:cdoptupper} is solved by
\begin{equation*}
\hat{\bm{\nu}}_k=T_{\bar{c}_k}\left(\tilde{\bm{\nu}}_k-\frac{1}{\bar{c}_k}\nabla_k L(\tilde{\bm{\nu}});\lambda_{0k},\lambda_{1k}\right).
\end{equation*}
\end{proposition}
A proof of Proposition~\ref{prop:cdupdate} is available in Appendix~\ref{app:propcdupdate}. The proposition states that a minimizer is given by appropriately thresholding a gradient descent update to coordinate group $k$. For both square and logistic loss, the gradient $\nabla_kL(\tilde{\bm{\nu}})$ can be expressed as
\begin{equation*}
\nabla_kL(\tilde{\bm{\nu}})=-\mathbf{X}_k^\top\mathbf{r},
\end{equation*}
where $\mathbf{r}=\mathbf{y}-\sum_j\mathbf{X}_j\tilde{\bm{\nu}}_j$ for square loss and $\mathbf{r}=\mathbf{y}-(1+\exp(-\sum_j\mathbf{X}_j\tilde{\bm{\nu}}_j))^{-1}$ for logistic loss. Hence, a solution to \eqref{eq:cdoptupper} can be computed in as few as $O(p_kn)$ operations.

Algorithm~\ref{alg:cd} now presents the coordinate descent scheme.
\begin{algorithm}[ht]
\caption{Coordinate descent}
\label{alg:cd}
\DontPrintSemicolon
\SetKw{Break}{break}
\SetKwInOut{Input}{input}
\Input{$\bm{\nu}^{(0)}\in\mathbb{R}^{\sum_kp_k}$}
\For{$m=1,2,\ldots$}{
$\bm{\nu}^{(m)}\gets\bm{\nu}^{(m-1)}$ \;
\For{$k=1,\ldots,g$}{
$\bm{\nu}_k^{(m)}\gets\underset{\bm{\xi}\in\mathbb{R}^{p_k}}{\argmin}~\bar{F}_{\bar{c}_k}(\bm{\nu}_1^{(m)},\ldots,\bm{\nu}_{k-1}^{(m)},\bm{\xi},\bm{\nu}_{k+1}^{(m)},\ldots,\bm{\nu}_g^{(m)};\bm{\nu}^{(m)})$ \;
}
\lIf{converged}{\Break}
}
\Return{$\bm{\nu}^{(m)}$}
\end{algorithm}
Several algorithmic optimizations and heuristics can improve performance, including gradient screening, gradient ordering, and active set updates. These are discussed in Appendix~\ref{app:tricks}.

While Algorithm~\ref{alg:cd} may appear as an otherwise standard coordinate descent algorithm, the presence of the group subset penalty complicates the analysis of its convergence properties. No standard convergence results directly apply. Hence, we work towards establishing a convergence result tailored to Algorithm~\ref{alg:cd}. For this task, we require the notion of a stationary point. Recall the lower directional derivative of a function $g:\mathbb{R}^p\to\mathbb{R}$ at a point $\mathbf{x}\in\mathbb{R}^p$ along the direction $\mathbf{v}\in\mathbb{R}^p$ is given by $\nabla_{\mathbf{v}}^-g(\mathbf{x}):=\liminf_{h\downarrow0}h^{-1}[g(\mathbf{x}+\mathbf{v}h)-g(\mathbf{x})]$. The definition for a stationary point now follows.
\begin{definition}
\label{def:statpoint}
A point $\bm{\nu}^\star\in\mathbb{R}^{\sum_kp_k}$ is said to be a stationary point of $F(\bm{\nu})$ if the lower directional derivative $\nabla_\mathbf{v}^-F(\bm{\nu}^\star)$ is nonnegative along all vectors $\bm{\mathbf{v}}\in\mathbb{R}^{\sum_kp_k}$.
\end{definition}
Definition~\ref{def:statpoint} is a standard notion of stationarity in the context of nondifferentiable functions \parencite{Tseng2001,Hazimeh2020}. With this basic definition at hand, we now introduce the stronger notion of a coordinate descent minimum point.
\begin{definition}
\label{def:cdmin}
A point $\bm{\nu}^\star\in\mathbb{R}^{\sum_kp_k}$ with active set $\mathcal{A}=\gs(\bm{\nu}^\star)$ is said to be a coordinate descent minimum point of $F(\bm{\nu})$ if it is a stationary point and
\begin{itemize}
\item For all $k\in\mathcal{A}$, it holds $\|\bm{\nu}_k^\star\|\geq\sqrt{2\lambda_{0k}/\bar{c}_k}$; and
\item For all $k\not\in\mathcal{A}$, it holds $\left(\|\nabla_kL(\bm{\nu}^\star)\|-\lambda_{1k}\right)/\bar{c}_k\leq\sqrt{2\lambda_{0k}/\bar{c}_k}$.
\end{itemize}
\end{definition}
The inequalities in Definition~\ref{def:cdmin}, which place additional restrictions on the norms of the group coefficients or gradients, follow from the thresholding function \eqref{eq:threshold}. All in all, a coordinate descent minimum point cannot be improved by partially minimizing in the directions of any coordinate group. Theorem~\ref{thrm:cdconverge} now establishes convergence of Algorithm~\ref{alg:cd} to such a point.
\begin{theorem}
\label{thrm:cdconverge}
Let $\bar{c}_k>c_k$ for all $k=1,\ldots,g$. Then the sequence of iterates $\{\bm{\nu}^{(m)}\}_{m\in\mathbb{N}}$ produced by Algorithm~\ref{alg:cd} converges to a fixed support $\mathcal{A}$ in finitely many iterations. Furthermore, if $\lambda_{1k}>0$ for all $k=1,\ldots,g$, or no elements of $\bm{\nu}$ tend to $\pm\infty$, then $\{\bm{\nu}^{(m)}\}_{m\in\mathbb{N}}$ converges to a coordinate descent minimum point $\bm{\nu}^\star$ of $F(\bm{\nu})$ with $\gs(\bm{\nu}^\star)=\mathcal{A}$.
\end{theorem}
To prove Theorem~\ref{thrm:cdconverge}, we follow a strategy similar to \textcite{Dedieu2021} and establish that the algorithm yields a sufficient decrease to the objective value and, from this property, show that the active set of the iterates eventually stabilizes. We can then treat the group subset penalty as a fixed quantity after some finite number of iterations and, in turn, call on existing convergence results. The full proof is in Appendix~\ref{app:thrmcdconverge}.

\subsection{Local search}

Local search algorithms have a long history in combinatorial optimization. Here we present one tailored specifically to the group subset problem. The proposed scheme generalizes an algorithm that first appeared in \textcite{Beck2013} for solving instances of unstructured sparse optimization. \textcite{Hazimeh2020,Dedieu2021} adapt it to best subset with promising results. The core idea is simple: given an incumbent solution, search a neighborhood local to that solution for a minimizer with lower objective value by discretely optimizing over a small set of coordinate groups. This scheme turns out to be useful when the predictors are strongly correlated, a situation in which coordinate descent alone may produce a poor solution.

Define the constraints sets
\begin{equation*}
C_s^1(\bm{\nu}):=\left\{\mathbf{z}\in\{0,1\}^{\sum_kp_k}:\operatorname{gs}(\mathbf{z})\subseteq\operatorname{gs}(\bm{\nu}),\sum_k1(\|\mathbf{z}_k\|\neq0)\leq s\right\}
\end{equation*}
and
\begin{equation*}
C_s^2(\bm{\nu}):=\left\{\mathbf{z}\in\{0,1\}^{\sum_kp_k}:\operatorname{gs}(\mathbf{z})\not\subseteq\operatorname{gs}(\bm{\nu}),\sum_k1(\|\mathbf{z}_k\|\neq0)\leq s\right\}.
\end{equation*}
Now, consider the optimization problem
\begin{equation}
\label{eq:ls}
\underset{\substack{\bm{\xi}\in\mathbb{R}^{\sum_kp_k} \\ \mathbf{z}^1\in C_s^1(\bm{\nu}),\,\mathbf{z}^2\in C_s^2(\bm{\nu})}}{\min}\,F(\bm{\nu}-\mathbf{z}^1\circ\bm{\nu}+\mathbf{z}^2\circ\bm{\xi}),
\end{equation}
where $\circ$ notates element-wise multiplication. Given a fixed vector $\bm{\nu}$, produced by Algorithm~\ref{alg:cd} say, a solution to \eqref{eq:ls} is a minimizer among all ways of replacing an subset of active coordinate groups in $\bm{\nu}$ with a previously inactive subset. The complexity of the problem is dictated by $s$, which controls the maximal size of these subsets. The limiting case $s=1$ admits an efficient computational scheme, described in Appendix~\ref{app:ls}. For other values of $s$, off-the-shelf mixed-integer optimizers are available (e.g., \texttt{CPLEX}, \texttt{Gurobi}, or \texttt{MOSEK}). We refer the interested reader to the discussion in \textcite{Hazimeh2020} for further details.

Algorithm~\ref{alg:ls} now presents the local search scheme.
\begin{algorithm}[ht]
\caption{Local search}
\label{alg:ls}
\DontPrintSemicolon
\SetKw{Break}{break}
\SetKwInOut{Input}{input}
\Input{$\hat{\bm{\nu}}^{(0)}\in\mathbb{R}^{\sum_kp_k}$}
\For{$m=1,2,\ldots$}{
$\bm{\nu}^{(m)}\gets$ result of running Algorithm~\ref{alg:cd} initialized with $\hat{\bm{\nu}}^{(m-1)}$ \;
$\hat{\bm{\nu}}^{(m)}\gets$ result of solving local search problem \eqref{eq:ls} with $\bm{\nu}=\bm{\nu}^{(m)}$ \;
\lIf{$F(\hat{\bm{\nu}}^{(m)})\not<F(\bm{\nu}^{(m)})$}{\Break}
}
\Return{$\bm{\nu}^{(m)}$}
\end{algorithm}
To summarize, the algorithm first produces a candidate solution using coordinate descent and then follows up by solving the local search problem \eqref{eq:ls}. This scheme is iterated until the solution cannot be improved. In the experiments to come, we solve the local search problem with subset sizes $s=1$  at each iteration of the loop. Compared with coordinate descent alone, this scheme can yield significantly lower objective values in high-correlation scenarios.

To characterize the convergence of Algorithm~\ref{alg:ls} and the quality of its iterates, we introduce the notion of a local search minimum point.
\begin{definition}
\label{def:lsmin}
A point $\bm{\nu}^\star\in\mathbb{R}^{\sum_kp_k}$ is said to be a local search minimum point of $F(\bm{\nu})$ if it is a coordinate descent minimum point and satisfies the inequality
\begin{equation*}
F(\bm{\nu}^\star)\leq\underset{\substack{\bm{\xi}\in\mathbb{R}^{\sum_kp_k} \\ \mathbf{z}^1\in C_s^1(\bm{\nu}^\star),\,\mathbf{z}^2\in C_s^2(\bm{\nu}^\star)}}{\min}\,F(\bm{\nu}^\star-\mathbf{z}^1\circ\bm{\nu}^\star+\mathbf{z}^2\circ\bm{\xi}).
\end{equation*}
\end{definition}
Definition~\ref{def:lsmin} states that local search minimum points are coordinate descent minimum points that are unable to be improved by replacing any set of $s$ coordinate groups. Theorem~\ref{thrm:lsconverge} now states that Algorithm~\ref{alg:ls} converges to such a point.
\begin{theorem}
\label{thrm:lsconverge}
Suppose the conditions of Theorem~\ref{thrm:cdconverge} hold. Then the sequence of iterates $\{\bm{\nu}^{(m)}\}_{m\in\mathbb{N}}$ produced by Algorithm~\ref{alg:ls} converges a local search minimum point $\bm{\nu}^\star$ of $F(\bm{\nu})$ in finitely many iterations.
\end{theorem}
Theorem~\ref{thrm:lsconverge} indicates that local search can lead to higher-quality solutions than coordinate descent due to the additional optimality condition imposed. A proof of the theorem is available in Appendix~\ref{app:thrmlsconverge}.

\subsection{Regularization sequence}

To ensure larger groups are not unfairly penalized more strongly than smaller groups, the parameters $\lambda_{0k}$ and $\lambda_{1k}$ are configured to reflect the group size $p_k$. Suitable default choices are $\lambda_{0k}=p_k\lambda_0$ and $\lambda_{1k}=\sqrt{p}_k\lambda_1$, where $\lambda_0$ and $\lambda_1$ are nonnegative. For fixed $\lambda_1$, we take a sequence $\{\lambda_0^{(t)}\}_{t=1}^{T}$ such that $\lambda_0^{(1)}$ yields $\hat{\bm{\nu}}=\mathbf{0}$, and sequentially warm start the algorithms. That is, the solution for $\lambda_0^{(t+1)}$ is obtained by using the solution from $\lambda_0^{(t)}$ as an initialization point. The sequence $\{\lambda_0^{(t)}\}_{t=1}^{T}$ is computed in such a way that the active set of groups corresponding to $\lambda_0^{(t+1)}$ is always different to that corresponding to $\lambda_0^{(t)}$. Appendix~\ref{app:lambdaseq} presents the details of this method.

\subsection{Explicit sparsity constraints}

A drawback to regularizing the number of selected groups by penalty rather than by constraint is that some sparsity levels may not be achieved along the sequence $\{\lambda_0^{(t)}\}_{t=1}^T$, i.e., there may be no value of $\lambda_0$ such that the sparsity level is some integer $s$. This gap between penalty and constraint arises due to the group subset regularizer being nonconvex. Unfortunately, when the regularizer is expressed as a constraint, the problem is not coordinate-wise separable and not amenable to coordinate descent. To this end, Appendix~\ref{app:proximal} describes an alternative proximal gradient descent approach that handles the regularizer in constraint form.

\section{Error bounds}
\label{sec:theory}

This section presents a finite-sample analysis of the proposed estimators. In particular, we state probabilistic upper bounds for the error of estimating the underlying regression function. These bounds accommodate overlapping groups and model misspecification. The role of structure and shrinkage is discussed, and comparisons are made with known bounds for other estimators.

\subsection{Setup}

The data is assumed to be generated according to the regression model
\begin{equation*}
y_i=f^0(\mathbf{x}_i)+\varepsilon_i,\quad i=1,\ldots,n,
\end{equation*}
where $f^0:\mathbb{R}^p\to\mathbb{R}$ is a regression function, $\mathbf{x}_i\in\mathbb{R}^p$ are fixed predictors, and $\varepsilon_i\sim\mathcal{N}(0,\sigma^2)$ is iid stochastic noise. This flexible specification encompasses the semiparametric model $f^0(\mathbf{x})=\sum_jf_j(x_j)$ (with $f_j$ zero, linear, or nonlinear) that is the focus of our empirical studies, and the linear model $f^0(\mathbf{x})=\mathbf{x}^\top\bm{\beta}^0$. Let $\mathbf{f}^0:=(f^0(\mathbf{x}_1),\ldots,f^0(\mathbf{x}_n))^\top$ be the vector of function evaluations at the sample points. The goal of this section is to place probabilistic upper bounds on $\|\mathbf{f}^0-\hat{\mathbf{f}}\|^2/n$, the estimation error of $\hat{\mathbf{f}}:=\mathbf{X}\hat{\bm{\beta}}$.

The objects of our analysis are the group subset estimators~\eqref{eq:overshrinkgroupsubset}. We allow the predictor groups $\mathcal{G}_1,\ldots,\mathcal{G}_g$ to overlap. To facilitate comparisons against existing results, we constrain the number of nonzero groups rather than penalize them. To this end, let $\mathcal{V}(s)$ be the set of all $\bar{\bm{\nu}}$ such that at most $s$ groups are nonzero:\footnote{For all values of the group subset penalty parameter $\lambda_0$, there exists a constraint parameter $s$ which yields an identical solution.}
\begin{equation*}
\mathcal{V}(s):=\left\{\bar{\bm{\nu}}\in\mathcal{V}:\sum_k1(\|\bar{\bm{\nu}}_k\|\neq0)\leq s\right\}.
\end{equation*}
We consider the regular group subset estimator:
\begin{equation}
\label{eq:groupsubsetcon}
\underset{\substack{\bm{\beta}\in\mathbb{R}^p,\,\bar{\bm{\nu}}\in\mathcal{V}(s) \\ \bm{\beta}=\sum_k\bar{\bm{\nu}}_k}}{\min}\,\frac{1}{n}\|\mathbf{y}-\mathbf{X}\bm{\beta}\|^2,
\end{equation}
and the shrinkage estimator:
\begin{equation}
\label{eq:shrinkgroupsubsetcon}
\underset{\substack{\bm{\beta}\in\mathbb{R}^p,\,\bar{\bm{\nu}}\in\mathcal{V}(s) \\ \bm{\beta}=\sum_k\bar{\bm{\nu}}_k}}{\min}\,\frac{1}{n}\|\mathbf{y}-\mathbf{X}\bm{\beta}\|^2+2\sum_k\lambda_k\|\bar{\bm{\nu}}_k\|.
\end{equation}
The results derived below apply to global minimizers of these nonconvex problems. The algorithms of the preceding section cannot guarantee such minimizers in general.\footnote{In recent work, \textcite{Guo2021} show that statistical properties of best subset remain valid when the attained minimum is within a certain neighborhood of the global minimum. We expect their analysis extends to structured settings.} If global optimality is of foremost concern, the output of our algorithms can be used to initialize a mixed-integer optimizer which can guarantee a global solution at additional computational expense.

\subsection{Bound for group subset selection}

We begin with Theorem~\ref{thrm:groupsubsetbound}, which characterizes an upper bound for group subset with no shrinkage. The notation $p_\mathrm{max}:=\max_kp_k$ represents the maximal group size. As is customary, we absorb numerical constants into the term $C>0$.
\begin{theorem}
\label{thrm:groupsubsetbound}
Let $\delta\in(0,1]$ and $\alpha\in(0,1)$. Then, for some numerical constant $C>0$, the group subset estimator~\eqref{eq:groupsubsetcon} satisfies
\begin{equation}
\label{eq:groupsubsetbound}
\frac{1}{n}\|\mathbf{f}^0-\hat{\mathbf{f}}\|^2\leq\underset{\substack{\bm{\beta}\in\mathbb{R}^p,\,\bar{\bm{\nu}}\in\mathcal{V}(s) \\ \bm{\beta}=\sum_k\bar{\bm{\nu}}_k}}{\min}\,\frac{1+\alpha}{(1-\alpha)n}\|\mathbf{f}^0-\mathbf{X}\bm{\beta}\|^2+\frac{C\sigma^2}{\alpha(1-\alpha)n}\left[sp_\mathrm{max}+s\log\left(\frac{g}{s}\right)+\log(\delta^{-1})\right]
\end{equation}
with probability at least $1-\delta$.
\end{theorem}
Appendix~\ref{app:thrmgroupsubsetbound} provides a proof of Theorem~\ref{thrm:groupsubsetbound}. The first term on the right-hand side of \eqref{eq:groupsubsetbound} is the error incurred by the oracle in approximating $\mathbf{f}^0$ as $\mathbf{X}\bm{\beta}$. In general, this error is unavoidable in finite-dimensional settings. The three terms inside the brackets have the following interpretations. The first term is the cost of estimating $\bm{\beta}$; with $s$ active groups, there are at most $s\times p_\mathrm{max}$ parameters to estimate. The second term is the price of selection; it follows from an upper bound on the total number of group subsets. The third term controls the trade-off between the tightness of the bound and the probability it is satisfied. Finally, the scalar $\alpha$ appears in the bound due to the proof technique \parencite[as in, e.g.,][]{Rigollet2015}. When $\mathbf{f}^0=\mathbf{X}\bm{\beta}^0$, $\alpha$ need not appear. \textcite{Hazimeh2023} obtain a similar bound for $\mathbf{f}^0=\mathbf{X}\bm{\beta}^0$ in the case of equisized nonoverlapping groups. In the special case that all groups are singletons, \eqref{eq:groupsubsetbound} matches the well-known bound for best subset \parencite{Raskutti2011}.

Theorem~\ref{thrm:groupsubsetbound} confirms that group subset is preferable to best subset in structured settings. Consider the following example. Suppose we have $g$ groups each of size $p_0$ so that the total number of predictors is $p=g\times p_0$. It follows for group sparsity level $s$ that the ungrouped selection problem involves choosing $s\times p_0$ predictors. Accordingly, the ungrouped bound scales as $sp_0+sp_0\log(p/(sp_0))=sp_0+sp_0\log(g/s)$. On the other hand, the grouped bound scales as $sp_0+s\log(g/s)$, i.e., it improves by a factor $p_0$ of the logarithm term.

\subsection{Bounds for group subset selection with shrinkage}

We now establish bounds for group subset with shrinkage. The results are analogous to those established in \textcite{Mazumder2023} for best subset with shrinkage. Two results are given, a bound where the error decays as $1/\sqrt{n}$, and another where the error decays as $1/n$. Adopting standard terminology \parencite[e.g.,][]{Hastie2015}, the former bound is referred to as a ``slow rate'' and the latter bound as a ``fast rate.''

The slow rate is presented in Theorem~\ref{thrm:groupslowbound} and proved in Appendix~\ref{app:thrmgroupslowbound}.
\begin{theorem}
\label{thrm:groupslowbound}
Let $\delta\in(0,1]$. Let $\gamma_k$ be the maximal eigenvalue of the matrix $\mathbf{X}_k^\top\mathbf{X}_k/n$ and
\begin{equation*}
\lambda_k\geq\frac{\sqrt{\gamma_k}\sigma}{\sqrt{n}}\sqrt{p_k+2\sqrt{p_k\log(g)+p_k\log(\delta^{-1})}+2\log(g)+2\log(\delta^{-1})},\quad k=1,\ldots,g.
\end{equation*}
Then the group subset estimator~\eqref{eq:shrinkgroupsubsetcon} satisfies
\begin{equation}
\label{eq:groupslowbound}
\frac{1}{n}\|\mathbf{f}^0-\hat{\mathbf{f}}\|^2\leq\underset{\substack{\bm{\beta}\in\mathbb{R}^p,\,\bar{\bm{\nu}}\in\mathcal{V}(s) \\ \bm{\beta}=\sum_k\bar{\bm{\nu}}_k}}{\min}\,\frac{1}{n}\|\mathbf{f}^0-\mathbf{X}\bm{\beta}\|^2+4\sum_k\lambda_k\|\bar{\bm{\nu}}_k\|
\end{equation}
with probability at least $1-\delta$.
\end{theorem}
In the case of no overlap, Theorem~\ref{thrm:groupslowbound} demonstrates that the shrinkage estimator satisfies the same slow rate as group lasso \parencite[][Theorem~3.1]{Lounici2011}. An identical expression to \textcite{Lounici2011} for $\lambda_k$ can be stated here using a more intricate chi-squared tail bound in the proof. In the case of overlap, the same slow rate can be obtained for group lasso from \textcite[][Lemma~4]{Percival2012}.

The following assumption is required to establish the fast rate.
\begin{assumption}
\label{asmn:sparseeigen}
Let $s<\min(n/p_\mathrm{max},g)/2$. Then there exists a $\phi_{2s}>0$ such that
\begin{equation*}
\underset{\substack{\bm{\theta}\in\mathbb{R}^p,\,\bar{\bm{\nu}}\in\mathcal{V}(2s) \\ \bm{\theta}=\sum_k\bar{\bm{\nu}}_k\neq\mathbf{0}}}{\min}\,\frac{\|\mathbf{X}\bm{\theta}\|}{\sqrt{n}\sum_k\|\bar{\bm{\nu}}_k\|}\geq\phi_{2s}.
\end{equation*}
\end{assumption}
Assumption~\ref{asmn:sparseeigen} is satisfied provided no collection of $2s$ groups have linearly dependent columns in $\mathbf{X}$. This condition is a weaker version of the restricted eigenvalue condition used in \textcite{Lounici2011,Percival2012} for group lasso, which (loosely speaking) places additional restrictions on the correlations of the columns in $\mathbf{X}$.

The fast rate is presented in Theorem~\ref{thrm:groupfastbound} and proved in Appendix~\ref{app:thrmgroupfastbound}. The notation $\lambda_\mathrm{max}:=\max_k\lambda_k$ represents the maximal shrinkage parameter.
\begin{theorem}
\label{thrm:groupfastbound}
Let Assumption~\ref{asmn:sparseeigen} hold. Let $\delta\in(0,1]$ and $\alpha\in(0,1)$. Let $\lambda_1,\ldots,\lambda_g\geq0$. Then, for some numerical constant $C>0$, the group subset estimator~\eqref{eq:shrinkgroupsubsetcon} satisfies
\begin{equation}
\label{eq:groupfastbound}
\begin{split}
&\frac{1}{n}\|\mathbf{f}^0-\hat{\mathbf{f}}\|^2\leq\underset{\substack{\bm{\beta}\in\mathbb{R}^p,\,\bar{\bm{\nu}}\in\mathcal{V}(s) \\ \bm{\beta}=\sum_k\bar{\bm{\nu}}_k}}{\min}\,\frac{1+\alpha}{(1-\alpha)n}\|\mathbf{f}^0-\mathbf{X}\bm{\beta}\|^2 \\
&\hspace{1.5in}+\frac{C\sigma^2}{\alpha(1-\alpha)n}\left[sp_\mathrm{max}+s\log\left(\frac{g}{s}\right)+\log(\delta^{-1})\right]+\frac{C\lambda_\mathrm{max}^2}{\alpha(1-\alpha)\phi_{2s}^2}
\end{split}
\end{equation}
with probability at least $1-\delta$.
\end{theorem}
Theorem~\ref{thrm:groupfastbound} establishes that the shrinkage estimator achieves the bound of the regular estimator up to an additional term that depends on $\lambda_\mathrm{max}$ and $\phi_{2s}$. By setting the shrinkage parameters to zero, the dependence on these terms vanishes, and the bounds are identical.

Theorems~\ref{thrm:groupslowbound} and \ref{thrm:groupfastbound} together show that group subset with shrinkage does no worse than group lasso or group subset. This property is helpful because group lasso tends to outperform when the noise is high or the sample size is small, while group subset tends to outperform in the opposite situation. This empirical observation is consistent with the above bounds since the slow rate~\eqref{eq:groupslowbound} depends on $\sigma/\sqrt{n}$ while the fast rate~\eqref{eq:groupfastbound} depends on $\sigma^2/n$. Hence, \eqref{eq:groupslowbound} is typically the tighter of the two bounds for large $\sigma$ or small $n$.

\section{Implementation}
\label{sec:implementation}

Our estimators and the algorithms for their computation are implemented in \texttt{grpsel}, an \texttt{R} package available on the public repository \texttt{CRAN}. To facilitate high-performance computation, the coordinate descent and local search algorithms from Section~\ref{sec:algorithms} are written in \texttt{C++}, making \texttt{grpsel} as fast (or faster) than \texttt{grpreg} \parencite{Breheny2015}, \texttt{grpregOverlap} \parencite{Zeng2016}, and \texttt{gamsel} \parencite{Chouldechova2015}---existing state-of-the-art packages for structured sparsity. \texttt{grpsel} currently supports square and logistic loss functions with overlapping and nonoverlapping groups. In contrast to packages like \texttt{grpregOverlap}, which handle overlapping groups by replicating a predictor each time it appears in a new group, our package does not use replication, making it more memory-efficient.\footnote{See also the \texttt{Julia} package \texttt{ParProx} \parencite{Ko2021}, which provides tools for memory-efficient structured sparsity.} Our package also provides functionality for automatically orthogonalizing the groups (i.e., transforming the group matrix $\mathbf{X}_k$ such that $\mathbf{X}_k^\top\mathbf{X}_k=\mathbf{I}$), which can have certain statistical benefits \parencite{Simon2012}. Whereas \texttt{grpreg} and \texttt{grpregOverlap} require this orthogonalization to operate correctly, it remains purely optional in \texttt{grpsel}.

\section{Simulations}
\label{sec:simulations}

This section investigates the statistical and computational performance of the proposed estimators for sparse semiparametric modeling on synthetic data. They are compared against group lasso, group SCAD, and group MCP, the latter two estimators being group versions of the smoothly clipped absolute deviation penalty \parencite{Fan2001} and the minimax concave penalty \parencite{Zhang2010}. These penalties introduce a nonconvexity parameter that helps remove unwanted shrinkage. The group subset estimators are fit using \texttt{grpsel}. Group lasso, group SCAD, and group MCP are fit using \texttt{grpreg}. The range of tuning parameters for each estimator is detailed in Appendix~\ref{app:tuning}.

All estimators use the same approach of fitting a sparse semiparametric model via regression splines and overlapping groups as described in Section~\ref{sec:intro}. The group penalty parameters are scaled to control the trade-off between fitting $f_j$ as linear or nonlinear. For group subset penalty parameter $\lambda$, we set $\lambda_k=\lambda$ for $k$ a linear group and $\lambda_k=2\lambda$ for $k$ a nonlinear group. For group lasso penalty parameter $\lambda$, we set $\lambda_k=\sqrt{2}\lambda$ for nonlinear groups to achieve an equivalent penalization.

Unlike \texttt{grpsel}, \texttt{grpreg} does not natively support overlapping groups. Unfortunately, the package \texttt{grpregOverlap}, which extends \texttt{grpreg} to handle overlapping groups, is no longer maintained on \texttt{CRAN}. To this end, we reimplement the approach in \texttt{grpregOverlap} of expanding the predictor matrix $\mathbf{X}\in\mathbb{R}^{n\times p}$ by replicating a predictor each time it appears in a new group to get $\tilde{\mathbf{X}}\in\mathbb{R}^{n\times\sum_kp_k}=(\mathbf{X}_1,\ldots,\mathbf{X}_g)$. \texttt{grpreg} is then run on the expanded predictor matrix $\tilde{\mathbf{X}}$ with nonoverlapping groups. This approach draws on the work of \textcite{Jacob2009}.

\subsection{Simulation design}

We study regression and classification. For regression, the response is generated according to
\begin{equation*}
y_i=f^0(\mathbf{x}_i)+\varepsilon_i,\quad i=1,\ldots,n,
\end{equation*}
while, for classification, it is generated as
\begin{equation*}
y_i=
\begin{cases}
0 & \text{if}~f^0(\mathbf{x}_i)+\varepsilon_i<0 \\
1 & \text{if}~f^0(\mathbf{x}_i)+\varepsilon_i\geq0,
\end{cases}
\quad i=1,\ldots,n.
\end{equation*}
In both cases, $f^0(\mathbf{x})=\sum_jf_j^0(x_j)$ and $\varepsilon_i\sim\mathcal{N}(0,\sigma^2)$. The covariates $\mathbf{x}_i$ are treated fixed and constructed by (1) drawing samples iid from $\mathcal{N}(\mathbf{0},\bm{\Sigma})$, (2) applying the standard normal distribution function to produce uniformly distributed variables that conform to $\bm{\Sigma}$ \parencite[see][]{Falk1999}, and (3) min-max scaling to the interval $[-1,1]$. The correlation matrix $\bm{\Sigma}$ is defined elementwise as $\Sigma_{i,j}=\rho^{|i-j|}$. The number of covariates is 2,500. For regression, 50 of the these covariates are selected at random to be nonzero: 40 relate to the response linearly as $f(x)=x$ and 10 relate nonlinearly as $f(x)=\cos(\pi x)$, $f(x)=\sin(\pi x)$, or $f(x)=\exp(10 x)$. For classification, support recovery is more difficult, so for that task 10 covariates are linear and 5 are nonlinear. The function evaluations are scaled to mean zero and variance one so that all functions are on the same footing. Each of the 2,500 covariates are expanded using a cubic regression spline containing three knots at equispaced quantiles.\footnote{Placing knots at equispaced quantiles is a common knot selection rule \parencite[see, e.g.,][]{James2021} that has the benefit of making the spline more flexible in high-density regions while restricting flexibility in low-density regions where less information is available.} Four terms are in each spline so that the number of predictors $p=10,000$. The number of groups $g=5,000$, consisting of 2,500 linear groups and 2,500 nonlinear groups. The sample size $n=1,000$ is fixed and the noise parameter $\sigma$ is varied to alter the signal-to-noise ratio (SNR).

\subsection{Statistical performance}

For regression, we measure out-of-sample test loss by the mean square loss on a testing set and report it relative to that of the best-performing estimator: 
\begin{equation*}
\operatorname{Relative~test~loss}:=\frac{\operatorname{Mean~square~loss}-\operatorname{Mean~square~loss}^\star}{\operatorname{Mean~square~loss}^\star},
\end{equation*}
where $\star$ indicates the minimal mean square (test) loss among the five estimators considered. The best possible value of this metric is zero. For classification, we report the same metric but measured in terms of mean logistic loss. As a measure of sparsity, we report the number of fitted functions that are nonzero:
\begin{equation*}
\operatorname{Sparsity}:=\sum_j\hat{f}_j\not\equiv0.
\end{equation*}
Finally, as a measure of support recovery, we report the micro F1-score for the classification of linear and nonlinear functions:
\begin{equation*}
\operatorname{Support~recovery}:=\frac{2\cdot\text{True positives}}{2\cdot\text{True positives}+\text{False positives}+\text{False negatives}}.
\end{equation*}
The best possible value of this metric is one and the null value is zero. These metrics are all evaluated using tuning parameters that minimize mean square loss or mean logistic loss over a separate validation set of size $n$.

The metrics under consideration are aggregated over 30 simulations. The regression results are reported in Figure~\ref{fig:sam-square} and the classification results in Figure~\ref{fig:sam-logistic}. The solid points are averages and the error bars are standard errors.
\begin{figure}[ht]
\centering
\input{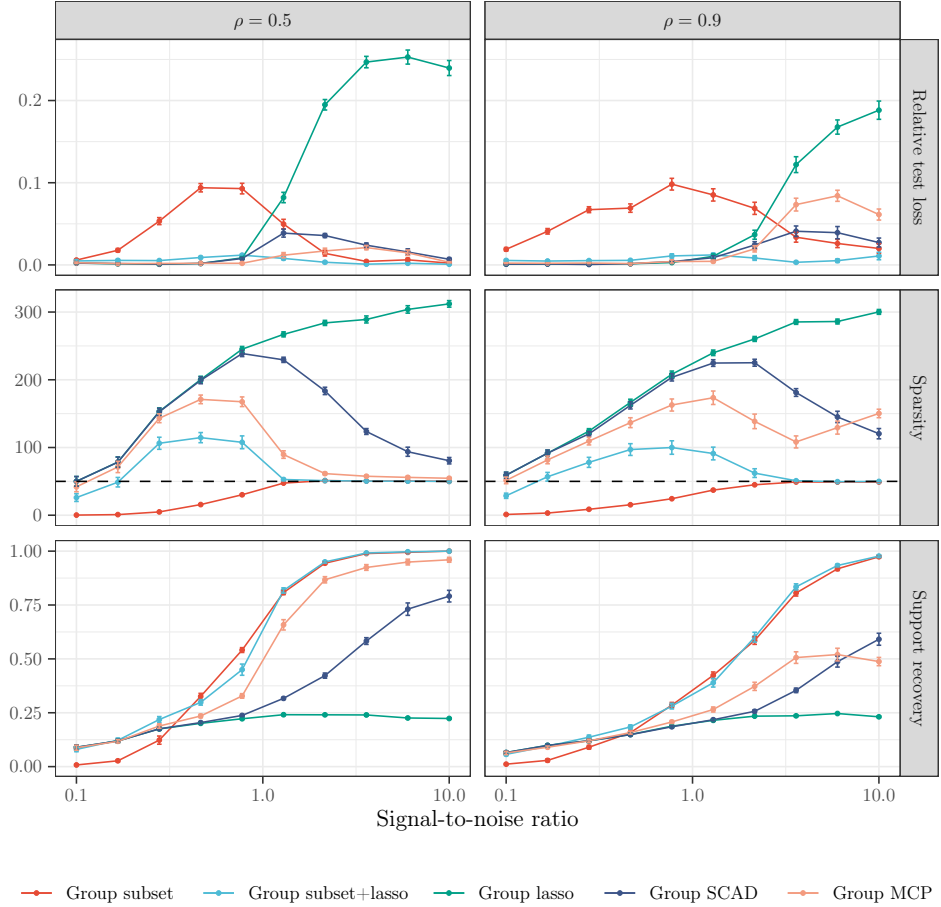}
\caption{Comparisons of estimators for sparse semiparametric regression. Metrics are aggregated over 30 synthetic datasets generated with $n=1,000$, $p=10,000$, and $g=5,000$. Solid points represent averages and error bars denote (one) standard errors. Dashed lines indicate the true number of nonzero functions.}
\label{fig:sam-square}
\end{figure}
\begin{figure}[ht]
\centering
\input{Figures/statistical-logistic.tex}
\caption{Comparisons of estimators for sparse semiparametric classification. Metrics are aggregated over 30 synthetic datasets generated with $n=1,000$, $p=10,000$, and $g=5,000$. Solid points represent averages and error bars denote (one) standard errors. Dashed lines indicate the true number of nonzero functions.}
\label{fig:sam-logistic}
\end{figure}

Consider first the regression results. In line with our theory, group subset exhibits excellent performance when the signal is strong but fares poorly when the signal is weak. Group lasso behaves contrarily, performing capably at low SNRs but poorly at high SNRs. The transition between the two estimators in terms of relative test loss occurs at $\operatorname{SNR}\approx2$ in the moderate-correlation scenario ($\rho=0.5$) and later at $\operatorname{SNR}\approx3$ in the high-correlation scenario ($\rho=0.9$). Group subset+shrinkage achieves the best of both worlds. It improves the performance of group subset when the signal is weak and, by tapering off shrinkage, eventually behaves like group subset when the signal is strong. Group SCAD and group MCP also try to unwind shrinkage via their nonconvexity parameters. Even so, there remains a gap between these estimators and group subset+shrinkage. The latter converges earlier to the right sparsity level and the correct model. The gap is most stark in the high-correlation scenario, where their support recovery is roughly half that of the group subset estimators at $\operatorname{SNR}=10$.

Turning our attention to the classification results, we again see that group lasso maintains an edge in prediction over group subset when the signal is weak and vice-versa when the signal is strong. As before, the gap between group lasso and group subset at low SNRs is closed once shrinkage is introduced. The group subset estimators also continue to have a clear upper hand in support recovery. Closer inspection reveals this upper hand is because they make fewer false positive selections than the competing estimators, similar to the regression setting. There is, however, one interesting difference to the regression setting: group subset+shrinkage has an advantage over group subset when the SNR is high. In this regime, it has noticeably lower relative test loss and marginally better support recovery. This result corresponds to a well-known phenomenon where the maximum likelihood estimator diverges as the classes become separable \parencite[see][]{Hastie2015}. Shrinkage has the desirable side-effect of preventing this divergence.

\subsection{Computational performance}

We now compare the computational performance of \texttt{grpsel} against \texttt{grpreg} for regression. Our estimators---group subset and group subset+shrinkage---and those of \texttt{grpreg}---group lasso, group SCAD, and group MCP---each solve different optimization problems, so it does not make sense to ask whether one implementation is faster than another for the same problem. Rather, the purpose of these comparisons is to provide indications of run time and computational complexity for alternative approaches to structured sparsity. Both \texttt{grpsel} and \texttt{grpreg} are set with a convergence tolerance of $10^{-4}$. All run times and iteration counts are measured with reference to the coordinate descent algorithms of each package over a grid of the primary tuning parameter. Where there is a secondary tuning parameter (e.g., the shrinkage parameter in \texttt{grpsel} or nonconvexity parameter in \texttt{grpreg}), the figures reported are averaged over the number of secondary parameter values evaluated. The simulation design is as before, but we now fix the SNR and vary the number of covariates to measure scalability.

The results as aggregated over 30 synthetic datasets are reported in Figure~\ref{fig:timings}. The vertical bars are averages and the error bars are standard errors.
\begin{figure}[ht]
\centering
\input{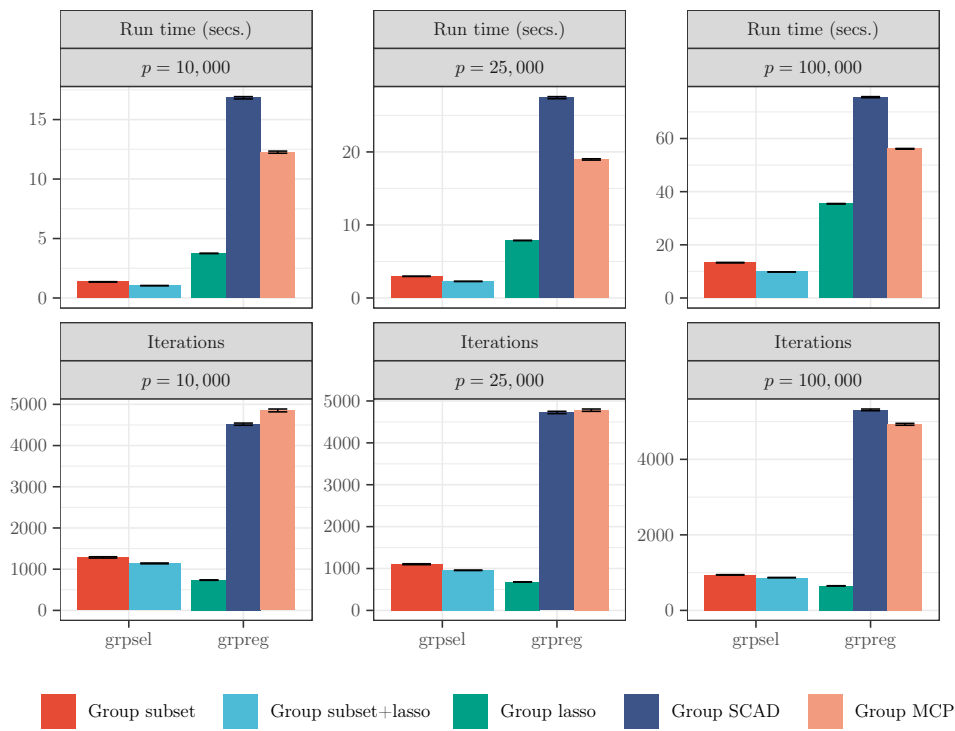}
\caption{Comparisons of \texttt{grpsel} and \texttt{grpreg}, and their estimators. Metrics are aggregated over 30 synthetic datasets generated with $\operatorname{SNR}=1$, $\rho=0.5$, and $n=1,000$. Vertical bars represent averages and error bars denote (one) standard errors.}
\label{fig:timings}
\end{figure}
For $p=10,000$, \texttt{grpsel} can compute an entire solution path in one to two seconds. Group subset+shrinkage achieves marginally lower run times than group subset, requiring slightly fewer iterations to converge thanks to the additional regularization from shrinkage. For $p=25,000$ the run times from \texttt{grpsel} are still less than five seconds to fit a path, while for $p=100,000$ the times are around 10 seconds. \texttt{grpreg} is also impressive in these scenarios, though relative to \texttt{grpsel} it is slower. Group lasso converges in the fewest iterations, followed by group subset+shrinkage. Group SCAD and group MCP always take several thousand iterations to converge.

Appendix~\ref{app:gamsel} reports additional computational experiments comparing \texttt{grpsel} with \texttt{gamsel} \parencite{Chouldechova2015}, an alternative package for sparse semiparametric modeling via overlapping groups and regression splines.

\section{Data analyses}
\label{sec:data}

This section studies two contemporary problems: modeling foot traffic in major supermarkets and modeling recessions in the business cycle. Both problems are characterized by the availability of many candidate predictors and the possibility for misspecification of linear models. These characteristics motivate consideration of sparse semiparametric models.

\subsection{Supermarket foot traffic}

The first dataset contains anonymized data on foot traffic and sales volumes for a major Chinese supermarket \parencite[see][]{Wang2009}.\footnote{Available at \url{https://personal.psu.edu/ril4/DataScience}.} The task is to model foot traffic using the sales volumes of different products. To facilitate managerial decision-making, the fitted model should identify a subset of products that well-predict foot traffic (i.e., it should be sparse).

The sample contains $n=464$ days. We randomly hold out 10\% of the data as a testing set and use the remaining data as a training set. Sales volumes are available for 6,398 products. A four-term cubic regression spline is used for each product, resulting in $p=25,592$ predictors and $g=12,796$ groups (6,398 linear groups plus 6,398 nonlinear groups). As a measure of predictive accuracy, we report mean square loss on the testing set. We also report the number of fitted functions that are nonzero. As benchmarks, we include random forest and lasso, which respectively produce dense nonparametric models and sparse linear models. In addition, the (unconditional) mean is evaluated as a predictive method to assess the value added by the predictors. Ten-fold cross-validation is used to choose tuning parameters.

The metrics under consideration are aggregated over 30 training-testing set splits and reported in Table~\ref{tab:market}. Averages are reported with standard errors in parentheses.
\begin{table}[ht]
\centering
\small
\begin{tabular}{lllll}
\toprule
 &  & \multicolumn{3}{c}{Sparsity} \\ 
\cmidrule(lr){3-5}
 & Mean square loss & Total & Linear & Nonlinear \\ 
\midrule
Group subset+shrinkage & 0.650 (0.026) & 178.6 (2.3) & 115.7 (1.2) & 62.9 (2.0) \\ 
Group lasso & 0.651 (0.025) & 267.0 (12.9) & 132.9 (2.3) & 134.1 (11.0) \\ 
Group MCP & 0.658 (0.026) & 215.1 (6.1) & 120.9 (1.2) & 94.1 (5.5) \\ 
Lasso & 0.712 (0.028) & 165.2 (5.0) & 165.2 (5.0) & - \\ 
Random forest & 1.330 (0.057) & - & - & - \\ 
Mean & 9.856 (0.367) & - & - & - \\ 
\bottomrule
\end{tabular}

\caption{Comparisons of methods for modeling supermarket foot traffic. Metrics are aggregated over 30 splits of the data into training and testing sets. Averages are reported next to (one) standard errors in parentheses.}
\label{tab:market}
\end{table}
Group subset+shrinkage used to fit sparse semiparametric models leads to the lowest mean square loss, though within statistical precision of group lasso and group MCP. Nonetheless, group subset+shrinkage has the edge of selecting fewer products than either of these estimators, which is advantageous for decision-making. Lasso yields models that are slightly more sparse but at the same time substantially less predictive, signifying linearity might be too restrictive. Random forest is markedly worse than any sparse method. It appears only a small fraction of products explain foot traffic, around 2--4\%.

\subsection{Economic recessions}

The second dataset contains monthly data on macroeconomic series for the United States \parencite[see][]{McCracken2016}.\footnote{The January 2022 vintage is used, available at \url{https://research.stlouisfed.org/econ/mccracken}. See Appendix~\ref{app:macro} for preprocessing steps.} The dataset is augmented with the National Bureau of Economic Research recession indicator, a dummy variable that indicates whether a given month is a period of recession or expansion.\footnote{Available at \url{https://fred.stlouisfed.org/series/USREC}.} The task is to model the recession indicator using the macroeconomic series. Such models are useful for scenario analysis and for assessing economic conditions in the absence of low-frequency variables such as quarterly gross domestic product growth.

The sample contains $n=746$ months, ending in October 2021. It includes the COVID-19 recession. We again randomly hold out 10\% of the data as a testing set and use the remaining data as a training set. Because there are relatively few recessionary periods, a stratified split is applied so that the proportion of recessions in the testing and training sets are equal. The dataset has 127 macroeconomic series. We add to this dataset six lags of each series, leading to a total set of 889 series. Applying a four-term spline to each series yields $p=3556$ predictors and $g=1778$ groups (889 linear groups plus 889 nonlinear groups). To evaluate predictive accuracy, we report mean logistic loss on the testing set. The remaining metrics and methods are as before.

The results as aggregated over 30 training-testing set splits are reported in Table~\ref{tab:macro}.
\begin{table}[ht]
\centering
\small
\begin{tabular}{lllll}
\toprule
 &  & \multicolumn{3}{c}{Sparsity} \\ 
\cmidrule(lr){3-5}
 & Mean logistic loss & Total & Linear & Nonlinear \\ 
\midrule
Group subset+shrinkage & 0.968 (0.103) & 44.0 (2.5) & 16.2 (0.9) & 27.8 (1.8) \\ 
Group lasso & 1.063 (0.102) & 57.5 (1.5) & 30.7 (0.7) & 26.8 (1.0) \\ 
Group MCP & 1.141 (0.092) & 31.5 (0.7) & 18.7 (0.5) & 12.8 (0.4) \\ 
Lasso & 1.093 (0.107) & 56.8 (1.1) & 56.8 (1.1) & - \\ 
Random forest & 1.294 (0.039) & - & - & - \\ 
Mean & 3.703 (0.000) & - & - & - \\ 
\bottomrule
\end{tabular}

\caption{Comparisons of methods for modeling economic recessions. Metrics are aggregated over 30 splits of the data into training and testing sets. Averages are reported next to (one) standard errors in parentheses.}
\label{tab:macro}
\end{table}
The sparse semiparametric models from group subset+shrinkage predict recessionary periods well. They perform better than those from group lasso and, at the same time, depend on fewer macroeconomic series. Group MCP is sparser still but also less predictive. Lasso also yields models worse than those from group lasso and group subset+shrinkage, highlighting the value in allowing for nonlinearity. As with the supermarket dataset, the dense models from random forest underperform relative to the sparse models from other methods. All methods improve on the mean, illustrating the predictive content of monthly series.

\section{Concluding remarks}
\label{sec:conclusion}

Despite a broad array of applications, structured sparsity via group subset selection is not well-studied, especially in high dimensions where it is computationally taxing. This paper represents an effort to close the gap. Our optimization framework consists of low complexity algorithms that come with convergence results. A theoretical analysis of the proposed estimators illuminates some of their finite-sample properties. The estimators behave favorably in simulation, exhibiting excellent support recovery when fitting sparse semiparametric models. In real-world modeling tasks, they improve on popular benchmarks.

Our implementation \texttt{grpsel} is available on the \texttt{R} repository \texttt{CRAN}.

\section*{Acknowledgments}

Ryan Thompson's research was supported by an Australian Government Research Training Program (RTP) Scholarship.

\printbibliography

\begin{appendices}

\section{Computation}

\subsection{Proof of Proposition~\ref{prop:cdupdate}}
\label{app:propcdupdate}

\begin{proof}
The subscript $k$ is dropped from $c_k$, $\lambda_{0k}$, and $\lambda_{1k}$ to simplify the notation. Since the objective is treated as a function in the $k$th group of coordinates $\bm{\nu}_k$ only, we have
\begin{equation*}
\begin{split}
\bar{F}_c(\bm{\nu};\tilde{\bm{\nu}})&\propto\nabla_k L(\tilde{\bm{\nu}})^\top(\bm{\nu}_k-\tilde{\bm{\nu}}_k)+\frac{c}{2}\|\bm{\nu}_k-\tilde{\bm{\nu}}_k\|^2+\lambda_01(\|\bm{\nu}_k\|\neq0)+\lambda_1\|\bm{\nu}_k\| \\
&\propto\frac{c}{2}\left\|\bm{\nu}_k-\left(\tilde{\bm{\nu}}_k-\frac{1}{c}\nabla_k L(\tilde{\bm{\nu}})\right)\right\|^2+\lambda_01(\|\bm{\nu}_k\|\neq0)+\lambda_1\|\bm{\nu}_k\| \\
&=\frac{c}{2}\left\|\bm{\nu}_k-\check{\bm{\nu}}_k\right\|^2+\lambda_01(\|\bm{\nu}_k\|\neq0)+\lambda_1\|\bm{\nu}_k\|,
\end{split}
\end{equation*}
where $\check{\bm{\nu}}_k=\tilde{\bm{\nu}}_k-1/c\nabla_kL(\tilde{\bm{\nu}})$. When $\lambda_1=0$, it is not hard to see a minimizer of $\bar{F}_c(\bm{\nu};\tilde{\bm{\nu}})$ is
\begin{equation}
\label{eq:proxsubset}
\hat{\bm{\nu}_k}=
\begin{dcases*}
\check{\bm{\nu}}_k & if $\|\check{\bm{\nu}}_k\|\geq\sqrt{\dfrac{2\lambda_0}{c}}$ \\
\mathbf{0} & otherwise.
\end{dcases*}
\end{equation}
When $\lambda_0=0$ and $\lambda_1>0$, the minimizer is
\begin{equation}
\label{eq:proxlasso}
\hat{\bm{\nu}}_k=
\begin{dcases*}
\left(1-\frac{\lambda_1}{c\|\check{\bm{\nu}}_k\|}\right)_+\check{\bm{\nu}}_k & if $\left(1-\dfrac{\lambda_1}{c\|\hat{\bm{\nu}}_k\|}\right)_+\|\check{\bm{\nu}}_k\|\geq0$ \\
\bm{0} & otherwise.
\end{dcases*}
\end{equation}
This expression follows from the proximal operator for the $l_2$-norm \parencite{Beck2017}. Combining \eqref{eq:proxsubset} with \eqref{eq:proxlasso} leads to the result of the proposition.
\end{proof}

\subsection{Coordinate descent heuristics}
\label{app:tricks}

We briefly outline the optimizations and heuristics used to accelerate our implementation of coordinate descent. These heuristics are similar to those used in the coordinate descent literature \parencite{Friedman2007,Breheny2011,Hazimeh2020}.

\subsubsection*{Gradient screening}

Rather than cycling through all $g$ groups in each coordinate descent round, it is convenient to restrict the updates to a smaller set of screened groups. The initialization $\bm{\nu}^{(0)}$ can be used to compute the coordinate-wise gradients $\{\|\nabla_kL(\bm{\nu}^{(0)})\|/\sqrt{p_k}\}_{k\not\in\mathcal{A}^{(0)}}$ which are already available as a consequence of selecting $\lambda_0$ dynamically (see Proposition~\ref{prop:lambdaseq}). The inactive groups whose gradients are among the top 500 largest are classed as ``strong,'' in addition to the active set of groups. The remaining groups are classed as ``weak.'' The coordinate descent updates are restricted to the strong groups until convergence is achieved, at which time a further round over the weak groups is performed. If the solution does not change after this further round, convergence is declared. Otherwise, any weak groups that have become active are shifted to the strong set, and the process is repeated.

\subsubsection*{Gradient ordering}

The solutions produced by coordinate descent often benefit from greedily ordering the groups. At the beginning of the algorithm, the groups are sorted according to their gradients. Any groups in $\mathcal{A}^{(0)}$ are placed first. The coordinate descent updates then proceed using this new ordering.

\subsubsection*{Active set updates}

The set of active groups typically stabilizes after several rounds of coordinate descent updates. At this time, several additional rounds are required for the nonzero coefficients to converge. Rather than cycling through the full set (or screened set) of groups, the updates are restricted to the active groups only, usually a small subset. Once convergence is achieved on the active set, a further round is performed over the inactive set to confirm overall convergence.

\subsection{Proof of Theorem~\ref{thrm:cdconverge}}
\label{app:thrmcdconverge}

The proof relies on the following lemma that characterizes the sequence of objective values from coordinate descent.

\begin{lemma}
\label{lemma:cdconverge}
Let $\bar{c}_k\geq c_k$ for all $k=1,\ldots,g$. Then the sequence of objective values $\{F(\bm{\nu}^{(m)})\}_{m\in\mathbb{N}}$ produced by Algorithm~\ref{alg:cd} is decreasing, convergent, and satisfies the inequality
\begin{equation*}
F(\bm{\nu}^{(m)})-F(\bm{\nu}^{(m+1)})\geq\sum_k\frac{\bar{c}_k-c_k}{2}\|\bm{\nu}_k^{(m+1)}-\bm{\nu}_k^{(m)}\|^2.
\end{equation*}
\begin{proof}
Denote by $\bm{\nu}^\star$ the result of applying the thresholding function~\eqref{eq:threshold} to $\tilde{\bm{\nu}}$. Starting from the inequality~\eqref{eq:upper} with $\bm{\nu}=\bm{\nu}^\star$, we add $R(\bm{\nu}^\star)$ to both sides to obtain
\begin{equation*}
F(\bm{\nu}^\star)\leq L(\tilde{\bm{\nu}})+\nabla_kL(\tilde{\bm{\nu}})^\top(\bm{\nu}_k^\star-\tilde{\bm{\nu}}_k)+\frac{c_k}{2}\|\bm{\nu}_k^\star-\tilde{\bm{\nu}}_k\|^2+R(\bm{\nu}^\star).
\end{equation*}
Adding $\bar{c}_k/2\|\bm{\nu}_k^\star-\tilde{\bm{\nu}}_k\|^2$ to both sides and rearranging terms leads to
\begin{equation*}
\begin{split}
F(\bm{\nu}^\star)&\leq L(\tilde{\bm{\nu}})+\nabla_kL(\tilde{\bm{\nu}})^\top(\bm{\nu}_k^\star-\tilde{\bm{\nu}}_k)+\frac{\bar{c}_k}{2}\|\bm{\nu}_k^\star-\tilde{\bm{\nu}}_k\|^2+R(\bm{\nu}^\star) +\frac{c_k-\bar{c}_k}{2}\|\bm{\nu}_k^\star-\tilde{\bm{\nu}}_k\|^2 \\
&=\bar{F}_{\bar{c}_k}(\bm{\nu}^\star;\tilde{\bm{\nu}})+\frac{c_k-\bar{c}_k}{2}\|\bm{\nu}_k^\star-\tilde{\bm{\nu}}_k\|^2.
\end{split}
\end{equation*}
Using $\bar{F}_{\bar{c}_k}(\bm{\nu}^\star;\tilde{\bm{\nu}})\leq\bar{F}_{\bar{c}_k}(\tilde{\bm{\nu}};\tilde{\bm{\nu}})=F(\tilde{\bm{\nu}})$, we reorganize terms to get
\begin{equation}
\label{eq:decrease}
F(\tilde{\bm{\nu}})-F(\bm{\nu}^\star)\geq \frac{\bar{c}_k-c_k}{2}\|\bm{\nu}_k^\star-\tilde{\bm{\nu}}_k\|^2.
\end{equation}
Now, define the vector
\begin{equation*}
\bm{\eta}^{(m)}_k:=
\begin{cases*}
(\bm{\nu}_1^{(m+1)\top},\ldots,\bm{\nu}_k^{(m+1)\top},\bm{\nu}_{k+1}^{(m)\top},\ldots,\bm{\nu}_g^{(m)\top})^\top & if $k>0$ \\
\bm{\nu}^{(m)} & otherwise.
\end{cases*}
\end{equation*}
Take $\tilde{\bm{\nu}}=\bm{\eta}_{k-1}^{(m)}$ and $\bm{\nu}^\star=\bm{\eta}_k^{(m)}$ and sum both sides of the inequality~\eqref{eq:decrease} over $1\leq k\leq g$ to get
\begin{equation*}
\sum_k[F(\bm{\eta}_{k-1}^{(m)})-F(\bm{\eta}_k^{(m)})]=F(\bm{\eta}_0^{(m)})-F(\bm{\eta}_g^{(m)})\geq\sum_k\frac{\bar{c}_k-c_k}{2}\|\bm{\nu}_k^{(m+1)}-\bm{\nu}_k^{(m)}\|^2.
\end{equation*}
By definition $\bm{\eta}_0^{(m)}=\bm{\nu}^{(m)}$ and $\bm{\eta}_g^{(m)}=\bm{\nu}^{(m+1)}$, establishing $\{F(\bm{\nu}^{(m)})\}_{m\in\mathbb{N}}$ is decreasing. Since $F(\bm{\nu})$ is bounded below, $\{F(\bm{\nu}^{(m)})\}_{m\in\mathbb{N}}$ must converge.
\end{proof}
\end{lemma}

We are now ready to prove Theorem~\ref{thrm:cdconverge}.

\begin{proof}
The proof is presented in three parts. First, we show the sequence of iterates $\{\bm{\nu}^{(m)}\}_{m\in\mathbb{N}}$ produced by Algorithm~\ref{alg:cd} stabilizes to a fixed support within a finite number of iterations. Second, we show these iterates converge to a stationary point $\bm{\nu}^\star$ per Definition~\ref{def:statpoint}. Third, we show $\bm{\nu}^\star$ is also a coordinate descent minimum point per Definition~\ref{def:cdmin}.

The first part is proven by contradiction along the lines of \textcite[][Theorem~1]{Dedieu2021}. Suppose the support does not stabilize in finitely many iterations. Choose an $m$ such that $\gs(\bm{\nu}^{(m+1)})\neq\gs(\bm{\nu}^{(m)})$. Then at least one group was added or removed from the support, i.e., there is a $k$ such that either (1) $\bm{\nu}_k^{(m)}=\mathbf{0}$ and $\bm{\nu}_k^{(m+1)}\neq\mathbf{0}$ or (2) $\bm{\nu}_k^{(m)}\neq\mathbf{0}$ and $\bm{\nu}_k^{(m+1)}=\mathbf{0}$. Consider case (1). It follows from Lemma~\ref{lemma:cdconverge}
\begin{equation*}
F(\bm{\nu}^{(m)})-F(\bm{\nu}^{(m+1)})\geq\frac{\bar{c}_k-c_k}{2}\|\bm{\nu}_k^{(m+1)}\|^2, 
\end{equation*}
and, because $\bm{\nu}_k^{(m+1)}$ is the output of the thresholding function~\eqref{eq:threshold}, it holds $\bm{\nu}_k^{(m+1)}\geq\sqrt{2\lambda_{0k}/\bar{c}_k}$. These inequalities together imply
\begin{equation*}
F(\bm{\nu}^{(m)})-F(\bm{\nu}^{(m+1)})\geq(\bar{c}_k-c_k)\frac{\lambda_{0k}}{\bar{c}_k}.
\end{equation*}
Similar working yields the same inequality for case (2). For $\bar{c}_k>c_k$, the quantity on the right-hand side is strictly positive. Hence, a change to the support yields a strict decrease in the objective value. However, if the support changes infinitely many times, this contradicts that $F(\bm{\nu})$ is bounded below. The support must, therefore, stabilize in a finite number of iterations.

We now turn to the second part of the proof that the iterates converge to a stationary point $\bm{\nu}^\star$. Since the active set stabilizes in finitely many iterations, there exists a finite $M$ such that the iterates of the sequence $\{\bm{\nu}^{(m)}\}_{m\geq M}$ share the same active set, say $\mathcal{A}$. Hence, for all $m\geq M$ and $k\in\mathcal{A}$, we have
\begin{equation*}
\bar{F}_{{\bar{c}}_k}(\bm{\nu};\bm{\nu}^{(m)})\propto\frac{\bar{c}_k}{2}\left\|\bm{\nu}_k-\left(\bm{\nu}_k^{(m)}-\frac{1}{\bar{c}_k}\nabla_kL(\bm{\nu}^{(m)})\right)\right\|^2+\lambda_{1k}\|\bm{\nu}_k\|.
\end{equation*}
Thus, the group subset penalty can be treated as fixed and the objective function reduces to that of the group lasso, which possesses the Kurdyka-Łojasiewicz property \parencite[see, e.g.,][]{Yu2022}. Furthermore, under the statement of the theorem, either (a) $\lambda_{1k}>0$ for all $k=1,\ldots,g$ (so the objective is coercive) or (b) no elements of $\bm{\nu}$ tend to $\pm\infty$. It follows then that the level set $\{\bm{\nu}\in\mathbb{R}^{\sum_kp_k}:F(\bm{\nu})\leq F(\bm{\nu}^{(0)})\}$ is bounded when the initialization $\bm{\nu}^{(0)}\in\mathbb{R}^{\sum_kp_k}$. Hence, by the descent property of Lemma~\ref{lemma:cdconverge}, the sequence $\{\bm{\nu}^{(m)}\}_{m\in\mathbb{N}}$ is bounded and therefore has a limit point $\bm{\nu}^\star$. These conditions are sufficient to invoke \textcite[][Theorems 1 and 2]{Xu2017} and conclude that $\bm{\nu}^\star$ is a stationary point of $F(\bm{\nu})$ with $\lim_{m\to\infty}\bm{\nu}^{(m)}=\bm{\nu}^\star$. This last result follows from the fact that our coordinate descent scheme falls in \citeauthor{Xu2017}'s (\citeyear{Xu2017}) ``deterministic block prox-linear'' algorithmic framework.

Finally, we show the third part of the proof that $\bm{\nu}^\star$ is also a coordinate descent minimum point. Let $M$ be such that the sequence $\{\bm{\nu}^{(m)}\}_{m\geq M}$ has fixed active set $\mathcal{A}$. A fixed active set implies that (a) the active groups $k\in\mathcal{A}$ survive thresholding for all $m\geq M$, and (b) the inactive groups $k\not\in\mathcal{A}$ are zero after thresholding for all $m\geq M$. Two inequalities then follow from inspection of the thresholding operator \eqref{eq:threshold}, both holding for $m\geq M$:
\begin{equation}
\label{eq:thlower}
\|\bm{\nu}_k^{(m)}\|\geq\sqrt{\frac{2\lambda_{0k}}{\bar{c}_k}},\quad\text{for all }k\in\mathcal{A},
\end{equation}
and
\begin{equation}
\label{eq:thupper}
\frac{\|\nabla_kL(\bm{\nu}^{(m)})\|-\lambda_{1k}}{\bar{c}_k}<\sqrt{\frac{2\lambda_{0k}}{\bar{c}_k}},\quad\text{for all }k\not\in\mathcal{A}.
\end{equation}
If these inequalities were not satisfied the active set would change from one iteration to another, contradicting support stabilization. Now, observing that the left-hand sides in \eqref{eq:thlower} and \eqref{eq:thupper} are continuous in $\bm{\nu}_k$, we take the limits as $m\to\infty$ to get
\begin{equation*}
\|\bm{\nu}_k^\star\|\geq\sqrt{\frac{2\lambda_{0k}}{\bar{c}_k}},\quad\text{for all }k\in\mathcal{A},
\end{equation*}
and
\begin{equation*}
\frac{\|\nabla_kL(\bm{\nu}^\star)\|-\lambda_{1k}}{\bar{c}_k}<\sqrt{\frac{2\lambda_{0k}}{\bar{c}_k}},\quad\text{for all }k\not\in\mathcal{A}.
\end{equation*}
The above inequalities match those in Definition~\ref{def:cdmin}. This result in combination with $\bm{\nu}^\star$ being a stationary point gives that it is also a coordinate descent minimum point of $F(\bm{\nu})$.
\end{proof}

\subsection{Local search solver}
\label{app:ls}

Algorithm~\ref{alg:solver1} solves the local search problem \eqref{eq:ls} when $s=1$.
\begin{algorithm}[ht]
\caption{Solver for \eqref{eq:ls} when $s=1$}
\label{alg:solver1}
\DontPrintSemicolon
\SetKwInOut{Input}{input}
\SetKw{Break}{break}
\Input{$\bm{\nu}\in\mathbb{R}^{\sum_kp_k}$}
$\mathcal{A}\gets\gs(\bm{\nu})$ \;
\For{$k\in\mathcal{A}$}{
\For{$j\not\in\mathcal{A}$}{
$\bm{\nu}^{(j)}\gets\bm{\nu}$ \;
$\bm{\nu}_k^{(j)}\gets\mathbf{0}$ \;
$\bm{\nu}_j^{(j)}\gets\underset{\bm{\xi}\in\mathbb{R}^{p_k}}{\argmin}~F(\bm{\nu}_1^{(j)},\ldots,\bm{\nu}_{j-1}^{(j)},\bm{\xi},\bm{\nu}_{j+1}^{(j)},\ldots,\bm{\nu}_g^{(j)})$ \;
}
$j^\star\gets\underset{j\not\in\mathcal{A}}{\argmin}~F(\bm{\nu}^{(j)})$ \;
\If{$F(\bm{\nu}^{(j^\star)})<F(\bm{\nu})$}{
$\bm{\nu}\gets\bm{\nu}^{(j^\star)}$ \;
\Break \;
}
}
\Return{$\bm{\nu}$}
\end{algorithm}
The algorithm comprises two low-complexity loops: an outer loop over the active set and an inner loop over the inactive set. Within the inner loop, an active coordinate group is removed, and the objective is minimized in the directions of an inactive coordinate group. This minimization problem can be solved by iterating the thresholding operator~\eqref{eq:threshold}. Though these loops are typically quite fast to iterate through, Algorithm~\ref{alg:ls} can be further accelerated via gradient screening. In our implementation, rather than searching through all inactive groups in the inner loop, only the inactive groups whose gradients are among the largest 5\% are enumerated.

\subsection{Proof of Theorem~\ref{thrm:lsconverge}}
\label{app:thrmlsconverge}

The proof follows that of \textcite[][Theorem~4]{Hazimeh2020}.

\begin{proof}
Suppose Algorithm~\ref{alg:ls} has not terminated by iteration $M$ and has produced the sequence of iterates $\{\bm{\nu}^{(m)}\}_{m=1}^M$. The corresponding sequence of objective values $\{F(\bm{\nu}^{(m)})\}_{m=1}^M$ must then be strictly decreasing otherwise the algorithm would have terminated. From Theorem~\ref{thrm:cdconverge}, Algorithm~\ref{thrm:cdconverge} converges to a fixed support and minimizes a convex function over this support. Hence, all the iterates must have different support otherwise the objective values would not decrease. However, there are only finitely many possible supports, so Algorithm~\ref{alg:ls} must terminate in finitely many iterations. Now, denote the solution at termination by $\bm{\nu}^\star$. To see that $\bm{\nu}^\star$ is a local search minimum point, it suffices to observe that Algorithm~\ref{alg:ls} terminates only when the coordinate descent minimum point from Algorithm~\ref{alg:cd} cannot be improved by solving the local search problem \eqref{eq:ls}. Hence, the conditions of Definition~\ref{def:lsmin} are satisfied.
\end{proof}

\subsection{Regularization sequence}
\label{app:lambdaseq}

Proposition~\ref{prop:lambdaseq} describes our approach to computing the sequence of regularization parameters $\{\lambda_0^{(t)}\}_{t=1}^T$ given fixed $\lambda_1$. The approach extends an idea of \textcite{Hazimeh2020} for best subset.
\begin{proposition}
\label{prop:lambdaseq}
Suppose that $\hat{\bm{\nu}}^{(t)}$ is the result of running Algorithm~\ref{alg:cd} with $\lambda_0=\lambda_0^{(t)}$. Let $\mathcal{A}^{(t)}=\gs(\hat{\bm{\nu}}^{(t)})$ be the active set of groups. Then running Algorithm~\ref{alg:cd} initialized to $\hat{\bm{\nu}}^{(t)}$ and using $\lambda_0=\lambda_0^{(t+1)}$ where
\begin{equation*}
\lambda_0^{(t+1)}=\alpha\cdot\underset{k\not\in\mathcal{A}^{(t)}}{\max}\left(\frac{\left(\|\nabla_kL(\hat{\bm{\nu}}^{(t)})\|-\lambda_{1k}\right)_+^2}{2p_k\bar{c}_k}\right)
\end{equation*}
produces a solution $\hat{\bm{\nu}}^{(t+1)}$ such that $\hat{\bm{\nu}}^{(t+1)}\neq\hat{\bm{\nu}}^{(t)}$ for any $\alpha\in[0,1)$.
\begin{proof}
Under the conditions of Theorem~\ref{thrm:cdconverge}, Algorithm~\ref{alg:cd} is guaranteed to converge to a stationary point $\hat{\bm{\nu}}^{(t)}$ such that for all $k\not\in\mathcal{A}^{(t)}$ it holds
\begin{equation*}
\frac{\left(\|\nabla_kL(\hat{\bm{\nu}}^{(t)})\|-\lambda_{1k}\right)_+}{\bar{c}_k}<\sqrt{\frac{2\lambda_{0k}^{(t)}}{\bar{c}_k}}=\sqrt{\frac{2\lambda_0^{(t)}p_k}{\bar{c}_k}}.
\end{equation*}
Then initializing Algorithm~\ref{alg:cd} with $\hat{\bm{\nu}}^{(t)}$ and using $\lambda_0=\lambda_0^{(t+1)}$ such that
\begin{equation*}
\lambda_0^{(t+1)}<\underset{k\not\in\mathcal{A}^{(t)}}{\max}\left(\frac{\left(\|\nabla_kL(\hat{\bm{\nu}}^{(t)})\|-\lambda_{1k}\right)_+^2}{2p_k\bar{c}_k}\right)
\end{equation*}
leads to $\hat{\bm{\nu}}^{(t+1)}\neq\hat{\bm{\nu}}^{(t)}$.
\end{proof}
\end{proposition}

\subsection{Proximal gradient descent}
\label{app:proximal}

The constrained version of the group subset problem \eqref{eq:overshrinkgroupsubset} shifts the group subset regularizer from the objective into the set of constraints:
\begin{equation}
\label{eq:overshrinkgroupsubsetcon}
\underset{\substack{\bm{\beta}\in\mathbb{R}^p,\,\bar{\bm{\nu}}\in\mathcal{V} \\ \bm{\beta}=\sum_{k}\bar{\bm{\nu}}_k \\ \sum_kw_k1(\|\bar{\bm{\nu}}_k\|\neq0)\leq s}}{\min}\,\sum_i\ell\left(\mathbf{x}_i^\top\bm{\beta},y_i\right)+\sum_k\lambda_{1k}\|\bar{\bm{\nu}}_k\|,
\end{equation}
where $w_k>0$ is a scaling factor that can, e.g., account for varying group sizes, similar to $\lambda_{0k}$ in \eqref{eq:overshrinkgroupsubset}. Along the same lines as the problem restatement in Section~\ref{sec:reformulation}, we can rewrite \eqref{eq:overshrinkgroupsubsetcon} in a more computationally friendly format involving only the latent coefficients:
\begin{equation*}
\underset{\substack{\bm{\nu}\in\mathbb{R}^{\sum_kp_k} \\ \sum_kw_k1(\|\bm{\nu}_k\|\neq0)\leq s}}{\min}\,\sum_i\ell\left(\sum_k\mathbf{x}_{ik}^\top\bm{\nu}_k,y_i\right)+\sum_k\lambda_{1k}\|\bm{\nu}_k\|.
\end{equation*}
Unlike its penalized cousin, the constrained group subset problem above is not amenable to coordinate descent as it is not coordinate-wise separable. We instead propose a proximal gradient descent \parencite{Beck2017} scheme, which iterates the update
\begin{equation}
\label{eq:pgdupdate}
\bm{\nu}^{(m)}\gets\underset{\substack{\bm{\nu}\in\mathbb{R}^{\sum_k p_k} \\ \sum_kw_k1(\|\bm{\nu}_k\|\neq0)\leq s}}{\operatorname{\arg\,\min}}\frac{c}{2}\left\|\bm{\nu}-\left(\bm{\nu}^{(m-1)}-\frac{1}{c}\nabla L(\bm{\nu}^{(m-1)})\right)\right\|^2+\sum_k\lambda_{1k}\|\bm{\nu}_k\|.
\end{equation}
Provided the step size $c$ is taken sufficiently large, convergence of the sequence of iterates $\{\bm{\nu}^{(m)}\}_{m\in\mathbb{N}}$ to a stationary point is readily established using standard arguments; see, e.g., the proof of Proposition~2 in \textcite{Mazumder2023}. Specifically, $c$ should be greater than $\gamma^2$ for square loss or $\gamma^2/4$ for logistic loss, where $\gamma$ is the maximal eigenvalue of $\mathbf{X}^\top\mathbf{X}$. The iterates can be initialized at the output of our coordinate descent or local search algorithms by taking $\bm{\nu}^{(0)}$ as the solution from these algorithms with sparsity nearest to $s$, generally leading to fast convergence.

The difficulty of the minimization in \eqref{eq:pgdupdate} depends on whether the scalars $w_k$ in the group subset constraint vary across $k$. Suppose first $w_k=1$ for all $k=1,\ldots,g$ and let $\hat{\bm{\nu}}=\bm{\nu}^{(m-1)}-1/c\nabla L(\bm{\nu}^{(m-1)})$. Then it is not difficult to show a closed form solution $\bm{\nu}^\star$ exists whose elements are given by
\begin{equation*}
\bm{\nu}_k^\star=
\begin{dcases*}
\left(1-\frac{\lambda_{1k}}{c\|\hat{\bm{\nu}}_k\|}\right)_+\hat{\bm{\nu}}_k & if $k\in\{(1),\ldots,(s)\}$ \\
\bm{0} & otherwise,
\end{dcases*}
\end{equation*}
where $\{(1),\ldots,(g)\}$ is an ordering of $\{1,\ldots,g\}$ such that
\begin{equation*}
\left\|\left(1-\frac{\lambda_{1(1)}}{c\|\hat{\bm{\nu}}_{(1)}\|}\right)_+\hat{\bm{\nu}}_{(1)}\right\|\geq\cdots\geq\left\|\left(1-\frac{\lambda_{1(g)}}{c\|\hat{\bm{\nu}}_{(g)}\|}\right)_+\hat{\bm{\nu}}_{(g)}\right\|.
\end{equation*}
The thresholding operation above is similar to that for the penalized problem, except now only the top $s$ groups are retained.

When the scaling factor $w_k$ varies across groups, the minimization \eqref{eq:pgdupdate} is a more difficult combinatorial problem, specifically a knapsack problem \parencite{Kellerer2004}. No closed-form solution is available in this case, but the problem is solvable by existing algorithms. For example, we can use a branch-and-bound approach as implemented in an off-the-shelf combinatorial optimizer. We need only rewrite the problem as a mixed-integer program:
\begin{equation}
\label{eq:projection}
\begin{split}
\underset{\substack{\bm{\nu}\in\mathbb{R}^{\sum_k p_k} \\ \mathbf{z}\in\{0,1\}^g}}{\operatorname{\min}}\,&\frac{c}{2}\|\bm{\nu}-\hat{\bm{\nu}}\|^2+\sum_k\lambda_{1k}\|\bm{\nu}_k\| \\
\operatorname{s.t.}\,\;\;\;&\sum_kw_kz_k\leq s \\
&(\|\bm{\nu}_k\|,1-z_k):\operatorname{SOS-1},\quad k=1,\ldots,g.
\end{split}
\end{equation}
Here, $\mathbf{z}=(z_1,\ldots,z_g)^\top$ is a vector of zero-one variables that define the set of selected groups, i.e., $\bm{\nu}_k$ is nonzero only if $z_k=1$. This relationship between $\bm{\nu}_k$ and $z_k$ is captured via the specially ordered set (SOS) constraint, which has the effect that
\begin{equation*}
(\|\bm{\nu}_k\|,1-z_k):\operatorname{SOS-1}\implies\|\bm{\nu}_k\|(1-z_k)=0, 
\end{equation*}
thereby forcing $\bm{\nu}_k=\mathbf{0}$ when $z_k=0$.

\section{Error bounds}

\subsection{Proof of Theorem~\ref{thrm:groupsubsetbound}}
\label{app:thrmgroupsubsetbound}

The proof requires the following lemma.

\begin{lemma}
\label{lemma:subsettailbound}
Let $\delta\in(0,1]$. Let $\mathbf{X}\in\mathbb{R}^{n\times p}$ be a fixed matrix and $\bm{\varepsilon}\in\mathbb{R}^n$ be a $\mathcal{N}(\mathbf{0},\sigma^2\mathbf{I})$ random vector. Define $\bm{\theta}:=\sum_k\bar{\bm{\nu}}_k$ and the random event
\begin{equation*}
A_{\bar{\bm{\nu}}}:=\left\{|\bm{\varepsilon}^\top\mathbf{X}\bm{\theta}|\geq\|\mathbf{X}\bm{\theta}\|C_1\sigma\sqrt{sp_\mathrm{max}+s\log\left(\frac{g}{s}\right)+\log(\delta^{-1})}\right\}.
\end{equation*}
Then, for some numerical constant $C_1>0$, the probability of the union $\cup_{\bar{\bm{\nu}}\in\mathcal{V}(2s)}A_{\bar{\bm{\nu}}}$ is at most $\delta$.
\end{lemma}

\begin{proof}
For $\mathcal{A}\subseteq\{1,\ldots,g\}$, a set of active groups, denote by $\mathcal{S}_\mathcal{A}:=\cup_{k\in\mathcal{A}}\mathcal{G}_k$ the set of active predictors. Denote the singular value decomposition of $\mathbf{X}_\mathcal{A}$ by $\mathbf{U}_\mathcal{A}\mathbf{D}_\mathcal{A}\mathbf{V}_\mathcal{A}^\top$, where $\mathbf{X}_\mathcal{A}$ are the columns of $\mathbf{X}$ indexed by $\mathcal{S}_\mathcal{A}$. Define the set of unit vectors $\mathcal{B}_2^r:=\{\mathbf{u}\in\mathbb{R}^r:\|\mathbf{u}\|\leq1\}$, and the set of $s$ group-sparse subsets $\mathcal{P}(s):=\{\mathcal{A}\subseteq\{1,\ldots,g\}:|\mathcal{A}|=s\}$. For all $\bar{\bm{\nu}}\in\mathcal{V}(2s)$ such that $\bm{\theta}\neq\bm{0}$, it holds
\begin{equation*}
\frac{|\bm{\varepsilon}^\top\mathbf{X}\bm{\theta}|}{\|\mathbf{X}\bm{\theta}\|}=\frac{|\bm{\varepsilon}^\top\mathbf{U}_\mathcal{A}\mathbf{D}_\mathcal{A}\mathbf{V}_\mathcal{A}^\top\bm{\theta}_\mathcal{A}|}{\|\mathbf{D}_\mathcal{A}\mathbf{V}_\mathcal{A}^\top\bm{\theta}_\mathcal{A}\|}\leq\max_{\mathcal{A}\in\mathcal{P}(2s)}\sup_{\mathbf{u}\in\mathcal{B}_2^{|\mathcal{S}_\mathcal{A}|}}|\bm{\varepsilon}^\top\mathbf{U}_\mathcal{A}\mathbf{u}|.
\end{equation*}
For any $t\in\mathbb{R}$, this inequality implies
\begin{equation*}
\p\left(\cup_{\bar{\bm{\nu}}\in\mathcal{V}(2s)}\left\{|\bm{\varepsilon}^\top\mathbf{X}\bm{\theta}|\geq\|\mathbf{X}\bm{\theta}\|t\right\}\right)\leq\p\left(\max_{\mathcal{A}\in\mathcal{P}(2s)}\sup_{\mathbf{u}\in\mathcal{B}_2^{|\mathcal{S}_\mathcal{A}|}}|\bm{\varepsilon}^\top\mathbf{U}_\mathcal{A}\mathbf{u}|\geq t\right).
\end{equation*}
Applying Boole's inequality to the right-hand side yields
\begin{equation*}
\p\left(\max_{\mathcal{A}\in\mathcal{P}(2s)}\sup_{\mathbf{u}\in\mathcal{B}_2^{|\mathcal{S}_\mathcal{A}|}}|\bm{\varepsilon}^\top\mathbf{U}_\mathcal{A}\mathbf{u}|\geq t\right)\leq\sum_{\mathcal{A}\in\mathcal{P}(2s)}\p\left(\sup_{\mathbf{u}\in\mathcal{B}_2^{|\mathcal{S}_\mathcal{A}|}}|\bm{\varepsilon}^\top\mathbf{U}_\mathcal{A}\mathbf{u}|\geq t\right).
\end{equation*}
We bound the supremum over $\mathcal{B}_2^{|\mathcal{S}_\mathcal{A}|}$ using an $\epsilon$-net argument. Let $\mathcal{E}^{|\mathcal{S}_\mathcal{A}|}$ be an $\epsilon$-net of $\mathcal{B}_2^{|\mathcal{S}_\mathcal{A}|}$ with respect to $l_2$-norm that satisfies $|\mathcal{E}^{|\mathcal{S}_\mathcal{A}|}|\leq(3/\epsilon)^{|\mathcal{S}_\mathcal{A}|}$. Such an $\mathcal{E}^{|\mathcal{S}_\mathcal{A}|}$ is guaranteed to exist for $\epsilon\in(0,1)$ \parencite[Lemma~1.18]{Rigollet2015}. Setting $\epsilon=1/2$, it holds for any $\mathcal{A}\in\mathcal{P}(2s)$ and any $\mathbf{z}\in\mathcal{E}^{|\mathcal{S}_\mathcal{A}|}$
\begin{equation*}
\sup_{\mathbf{u}\in\mathcal{B}_2^{|\mathcal{S}_\mathcal{A}|}}|\bm{\varepsilon}^\top\mathbf{U}_\mathcal{A}\mathbf{u}|\leq2\sup_{\mathbf{z}\in\mathcal{E}^{|\mathcal{S}_\mathcal{A}|}}|\bm{\varepsilon}^\top\mathbf{U}_\mathcal{A}\mathbf{z}|.
\end{equation*}
See the proof of Theorem~1.19 in \textcite{Rigollet2015} for a derivation of this bound. Applying Boole's inequality to the bound yields
\begin{equation*}
\sum_{\mathcal{A}\in\mathcal{P}(2s)}\p\left(2\sup_{\mathbf{z}\in\mathcal{E}^{|\mathcal{S}_\mathcal{A}|}}|\bm{\varepsilon}^\top\mathbf{U}_\mathcal{A}\mathbf{z}|\geq t\right)\leq\sum_{\mathcal{A}\in\mathcal{P}(2s)}\sum_{\mathbf{z}\in\mathcal{E}^{|\mathcal{S}_\mathcal{A}|}}\p\left(2|\bm{\varepsilon}^\top\mathbf{U}_\mathcal{A}\mathbf{z}|\geq t\right).
\end{equation*}
The cardinality of $\mathcal{P}(2s)$ satisfies $|\mathcal{P}(2s)|=\binom{g}{2s}\leq(eg/(2s))^{2s}$. For any $\mathcal{A}\in\mathcal{P}(2s)$, the cardinality of $\mathcal{E}^{|\mathcal{S}_\mathcal{A}|}$ satisfies $|\mathcal{E}^{|\mathcal{S}_\mathcal{A}|}|\leq6^{|\mathcal{S}_\mathcal{A}|}\leq6^{2sp_\mathrm{max}}$. Since $\mathbf{U}_\mathcal{A}$ is orthonormal and $\mathbf{z}$ has unit length, the random variable $\bm{\varepsilon}^\top\mathbf{U}_\mathcal{A}\mathbf{z}\sim\mathcal{N}(0,\sigma^2)$. Using a standard Gaussian tail bound \parencite[Lemma~1.4]{Rigollet2015}, we have
\begin{equation*}
\p\left(2|\bm{\varepsilon}^\top\mathbf{U}_\mathcal{A}\mathbf{z}|\geq t\right)\leq2\exp\left(-\frac{t^2}{8\sigma^2}\right).
\end{equation*}
It follows from the chain of inequalities above
\begin{equation*}
\p\left(\cup_{\bar{\bm{\nu}}\in\mathcal{V}(2s)}\left\{|\bm{\varepsilon}^\top\mathbf{X}\bm{\theta}|\geq\|\mathbf{X}\bm{\theta}\|t\right\}\right)\leq2\exp\left(-\frac{t^2}{8\sigma^2}+2sp_\mathrm{max}\log(6)+2s\log\left(\frac{eg}{2s}\right)\right).
\end{equation*}
Setting $t\geq\sqrt{8\sigma^2[\log(2)+2sp_\mathrm{max}\log(6)+2s\log(eg/(2s))+\log(\delta^{-1})]}$ concludes the proof.
\end{proof}

We are now ready to prove Theorem~\ref{thrm:groupsubsetbound}.

\begin{proof}
Take any $\bar{\bm{\nu}}\in\mathcal{V}(s)$ and any $\bm{\beta}\in\mathbb{R}^p$ such that $\bm{\beta}=\sum_k\bar{\bm{\nu}}_k$. Optimality of $\hat{\bar{\bm{\nu}}}$ and $\hat{\bm{\beta}}=\sum_k\hat{\bar{\bm{\nu}}}_k$ implies
\begin{equation*}
\frac{1}{n}\|\mathbf{y}-\mathbf{X}\hat{\bm{\beta}}\|^2\leq\frac{1}{n}\|\mathbf{y}-\mathbf{X}\bm{\beta}\|^2,
\end{equation*}
which, after some algebra, leads to
\begin{equation*}
\frac{1}{n}\|\mathbf{f}^0-\mathbf{X}\hat{\bm{\beta}}\|^2\leq\frac{1}{n}\|\mathbf{f}^0-\mathbf{X}\bm{\beta}\|^2+\frac{2}{n}|\bm{\varepsilon}^\top\mathbf{X}(\hat{\bm{\beta}}-\bm{\beta})|.
\end{equation*}
Observe $\hat{\bm{\beta}}-\bm{\beta}=\sum_k(\hat{\bar{\bm{\nu}}}_k-\bar{\bm{\nu}}_k)$, with at most $2s$ components of $(\hat{\bar{\bm{\nu}}}_{1}-\bar{\bm{\nu}}_1,\ldots,\hat{\bar{\bm{\nu}}}_g-\bar{\bm{\nu}}_g)$ not equal to $\mathbf{0}$. An application of Lemma~\ref{lemma:subsettailbound} thus yields the high-probability upper bound
\begin{equation*}
\begin{split}
\frac{2}{n}|\bm{\varepsilon}^\top\mathbf{X}(\hat{\bm{\beta}}-\bm{\beta})|&\leq\frac{2}{n}\|\mathbf{X}(\hat{\bm{\beta}}-\bm{\beta})\|C_1\sigma\sqrt{sp_\mathrm{max}+s\log\left(\frac{g}{s}\right)+\log(\delta^{-1})} \\
&\leq\frac{2}{n}\left(\|\mathbf{f}^0-\mathbf{X}\hat{\bm{\beta}}\|+\|\mathbf{f}^0-\mathbf{X}\bm{\beta}\|\right)C_1\sigma\sqrt{sp_\mathrm{max}+s\log\left(\frac{g}{s}\right)+\log(\delta^{-1})},
\end{split}
\end{equation*}
where the last line follows from Minkowski's inequality for $l_p$-norms. Using Young's inequality ($2ab\leq\alpha a^2+\alpha^{-1} b^2$ for $\alpha>0$), the first term on the right-hand side is bounded as
\begin{equation*}
\begin{split}
&\frac{2}{n}\|\mathbf{f}^0-\mathbf{X}\hat{\bm{\beta}}\|C_1\sigma\sqrt{sp_\mathrm{max}+s\log\left(\frac{g}{s}\right)+\log(\delta^{-1})} \\
&\hspace{2in}\leq\frac{\alpha}{n}\|\mathbf{f}^0-\mathbf{X}\hat{\bm{\beta}}\|^2+\frac{C_1^2\sigma^2}{\alpha n}\left[sp_\mathrm{max}+s\log\left(\frac{g}{s}\right)+\log(\delta^{-1})\right].
\end{split}
\end{equation*}
A bound for the second term on the right-hand side follows similarly. Putting the results together and rearranging terms, we arrive at
\begin{equation*}
\frac{1}{n}\|\mathbf{f}^0-\mathbf{X}\hat{\bm{\beta}}\|^2\leq\frac{1+\alpha}{(1-\alpha)n}\|\mathbf{f}^0-\mathbf{X}\bm{\beta}\|^2+\frac{2C_1^2\sigma^2}{\alpha(1-\alpha)n}\left[sp_\mathrm{max}+s\log\left(\frac{g}{s}\right)+\log(\delta^{-1})\right],
\end{equation*}
holding with probability at least $1-\delta$ for $\alpha\in(0,1)$. Taking $C\geq2C_1^2$ completes the proof.
\end{proof}

\subsection{Proof of Theorem~\ref{thrm:groupslowbound}}
\label{app:thrmgroupslowbound}

The proof requires the following lemma.

\begin{lemma}
\label{lemma:lassotailbound}
Let $\delta\in(0,1]$. Let $\mathbf{X}\in\mathbb{R}^{n\times p}$ be a fixed matrix and $\bm{\varepsilon}\in\mathbb{R}^n$ be a $\mathcal{N}(\mathbf{0},\sigma^2\mathbf{I})$ random vector. Let $\gamma_k$ be the maximal eigenvalue of $\mathbf{X}_k^\top\mathbf{X}_k/n$, where $\mathbf{X}_k$ is the submatrix of $\mathbf{X}$ corresponding to group $k$. Define the random event
\begin{equation*}
A_k=\left\{\|\mathbf{X}_k^\top\bm{\varepsilon}\|\geq\sqrt{n\gamma_k}\sigma\sqrt{p_k+2\sqrt{p_k\log(g)+p_k\log(\delta^{-1})}+2\log(g)+2\log(\delta^{-1})}\right\}.
\end{equation*}
Then the probability of the union $\cup_kA_k$ is at most $\delta$.
\begin{proof}
Denote the singular value decomposition of $\mathbf{X}_k$ by $\mathbf{U}_k\mathbf{D}_k\mathbf{V}_k^\top$. Using the properties of the operator norm, it holds
\begin{equation*}
\|\mathbf{X}_k^\top\bm{\varepsilon}\|=\|\mathbf{V}_k\mathbf{D}_k\mathbf{U}_k^\top\bm{\varepsilon}\|\leq\sqrt{n\gamma_k}\|\mathbf{U}_k^\top\bm{\varepsilon}\|. 
\end{equation*}
For any $t\in\mathbb{R}$, this inequality implies
\begin{equation*}
\p\left(\cup_k\left\{\|\mathbf{X}_k^\top\bm{\varepsilon}\|\geq\sqrt{n\gamma_k}t\right\}\right)\leq\p\left(\cup_k\left\{\|\mathbf{U}_k^\top\bm{\varepsilon}\|\geq t\right\}\right).
\end{equation*}
Applying Boole's inequality to the right-hand side yields
\begin{equation*}
\p\left(\cup_k\left\{\|\mathbf{U}_k^\top\bm{\varepsilon}\|\geq t\right\}\right)\leq\sum_k\p\left(\|\mathbf{U}_k^\top\bm{\varepsilon}\|\geq t\right).
\end{equation*}
Since $\mathbf{U}_k$ is orthonormal, the random variable $\|\mathbf{U}_k^\top\bm{\varepsilon}\|^2/\sigma^2\sim\chi^2(p_k)$. Using a standard chi-squared tail bound \parencite[Lemma~1]{Laurent2000}, we have for $t=p_k+\sqrt{2p_kx}+2x$ and $x>0$
\begin{equation*}
\p\left(\|\mathbf{U}_k^\top\bm{\varepsilon}\|\geq\sigma\sqrt{t}\right)\leq\exp(-x).
\end{equation*}
It follows from the chain of inequalities above
\begin{equation*}
\p\left(\cup_k\left\{\|\mathbf{X}_k^\top\bm{\varepsilon}\|\geq\sqrt{n\gamma_k}\sigma\sqrt{p_k+\sqrt{2p_kx}+2x}\right\}\right)\leq\exp(-x+\log(g)).
\end{equation*}
Setting $x\geq\log(g)+\log(\delta^{-1})$ concludes the proof.
\end{proof}
\end{lemma}

We are now ready to prove Theorem~\ref{thrm:groupslowbound}.

\begin{proof}
For any $\bar{\bm{\nu}}\in\mathcal{V}(s)$ and any $\bm{\beta}\in\mathbb{R}^p$ such that $\bm{\beta}=\sum_k\bar{\bm{\nu}}_k$, we have
\begin{equation*}
\frac{1}{n}\|\mathbf{y}-\mathbf{X}\hat{\bm{\beta}}\|^2+2\sum_k\lambda_k\|\hat{\bar{\bm{\nu}}}_k\|\leq\frac{1}{n}\|\mathbf{y}-\mathbf{X}\bm{\beta}\|^2+2\sum_k\lambda_k\|\bar{\bm{\nu}}_k\|,
\end{equation*}
which leads to
\begin{equation}
\label{eq:basicinequality}
\frac{1}{n}\|\mathbf{f}^0-\mathbf{X}\hat{\bm{\beta}}\|^2\leq\frac{1}{n}\|\mathbf{f}^0-\mathbf{X}\bm{\beta}\|^2+\frac{2}{n}|\bm{\varepsilon}^\top\mathbf{X}(\hat{\bm{\beta}}-\bm{\beta})|+2\sum_k\lambda_k\left(\|\bar{\bm{\nu}}_k\|-\|\hat{\bar{\bm{\nu}}}_k\|\right).
\end{equation}
The Cauchy-Schwarz inequality and Minkowski's inequality are applied in turn to get
\begin{equation*}
\frac{2}{n}|\bm{\varepsilon}^\top\mathbf{X}(\hat{\bm{\beta}}-\bm{\beta})|\leq\frac{2}{n}\sum_k\|\mathbf{X}_k^\top\bm{\varepsilon}\|\|\hat{\bar{\bm{\nu}}}_k-\bar{\bm{\nu}}_k\|\leq\frac{2}{n}\sum_k\|\mathbf{X}_k^\top\bm{\varepsilon}\|(\|\hat{\bar{\bm{\nu}}}_k\|+\|\bar{\bm{\nu}}_k\|).
\end{equation*}
Applying Lemma~\ref{lemma:lassotailbound}, and using the assumed lower bound on $\lambda_k$, yields
\begin{equation*}
\frac{2}{n}\sum_k\|\mathbf{X}_k^\top\bm{\varepsilon}\|(\|\hat{\bar{\bm{\nu}}}_k\|+\|\bar{\bm{\nu}}_k\|)\leq2\sum_k\lambda_k(\|\hat{\bar{\bm{\nu}}}_k\|+\|\bar{\bm{\nu}}_k\|)
\end{equation*}
with high-probability. Plugging this bound into \eqref{eq:basicinequality}, we arrive at
\begin{equation*}
\frac{1}{n}\|\mathbf{f}^0-\mathbf{X}\hat{\bm{\beta}}\|^2\leq\frac{1}{n}\|\mathbf{f}^0-\mathbf{X}\bm{\beta}\|^2+4\sum_k\lambda_k\|\bar{\bm{\nu}}_k\|,
\end{equation*}
holding with probability at least $1-\delta$.
\end{proof}

\subsection{Proof of Theorem~\ref{thrm:groupfastbound}}
\label{app:thrmgroupfastbound}

\begin{proof}
Begin with inequality~\eqref{eq:basicinequality}. First, we bound the term $2/n|\bm{\varepsilon}^\top\mathbf{X}(\hat{\bm{\beta}}-\bm{\beta})|$. Lemma~\ref{lemma:subsettailbound} gives the high-probability upper bound
\begin{equation*}
\begin{split}
\frac{2}{n}|\bm{\varepsilon}^\top\mathbf{X}(\hat{\bm{\beta}}-\bm{\beta})|&\leq\frac{2}{n}\|\mathbf{X}(\hat{\bm{\beta}}-\bm{\beta})\|C_1\sigma\sqrt{sp_\mathrm{max}+s\log\left(\frac{g}{s}\right)+\log(\delta^{-1})} \\
&\leq\frac{2}{n}\left(\|\mathbf{f}^0-\mathbf{X}\hat{\bm{\beta}}\|+\|\mathbf{f}^0-\mathbf{X}\bm{\beta}\|\right)C_1\sigma\sqrt{sp_\mathrm{max}+s\log\left(\frac{g}{s}\right)+\log(\delta^{-1})}.
\end{split}
\end{equation*}
Using Young's inequality ($2ab\leq\alpha/2a^2+2/\alpha b^2$), the first term on the right-hand side is bounded as
\begin{equation*}
\begin{split}
&\frac{2}{n}\|\mathbf{f}^0-\mathbf{X}\hat{\bm{\beta}}\|C_1\sigma\sqrt{sp_\mathrm{max}+s\log\left(\frac{g}{s}\right)+\log(\delta^{-1})} \\
&\hspace{1.5in}\leq\frac{\alpha}{2n}\|\mathbf{f}^0-\mathbf{X}\hat{\bm{\beta}}\|^2+\frac{2C_1^2\sigma^2}{\alpha n}\left[sp_\mathrm{max}+s\log\left(\frac{g}{s}\right)+\log(\delta^{-1})\right].
\end{split}
\end{equation*}
The second term on the right-hand side is bounded similarly. Now, we bound the remaining term $2\sum_k\lambda_k(\|\bar{\bm{\nu}}_k\|-\|\hat{\bar{\bm{\nu}}}_k\|)$ in \eqref{eq:basicinequality}. Minkowski's inequality and Assumption~\ref{asmn:sparseeigen} give
\begin{equation*}
2\sum_k\lambda_k(\|\bar{\bm{\nu}}_k\|-\|\hat{\bar{\bm{\nu}}}_k\|)\leq2\lambda_\mathrm{max}\sum_k\|\bar{\bm{\nu}}_k-\hat{\bar{\bm{\nu}}}_k\|\leq\frac{2\lambda_\mathrm{max}}{\sqrt{n}\phi_{2s}}\|\mathbf{X}(\bm{\beta}-\hat{\bm{\beta}})\|.
\end{equation*}
Using Minkowski's inequality again, we have
\begin{equation*}
\frac{2\lambda_\mathrm{max}}{\sqrt{n}\phi_{2s}}\|\mathbf{X}(\bm{\beta}-\hat{\bm{\beta}})\|\leq\frac{2\lambda_\mathrm{max}}{\sqrt{n}\phi_{2s}}\left(\|\mathbf{f}^0-\mathbf{X}\hat{\bm{\beta}}\|+\|\mathbf{f}^0-\mathbf{X}\bm{\beta}\|\right).
\end{equation*}
Two applications of Young's inequality yields
\begin{equation*}
\frac{2\lambda_\mathrm{max}}{\sqrt{n}\phi_{2s}}\left(\|\mathbf{f}^0-\mathbf{X}\hat{\bm{\beta}}\|+\|\mathbf{f}^0-\mathbf{X}\bm{\beta}\|\right)\leq\frac{4\lambda_\mathrm{max}^2}{\alpha\phi_{2s}^2}+\frac{\alpha}{2n}\|\mathbf{f}^0-\mathbf{X}\hat{\bm{\beta}}\|^2+\frac{\alpha}{2n}\|\mathbf{f}^0-\mathbf{X}\bm{\beta}\|^2.
\end{equation*}
Finally, putting the bounds together and simplifying the resulting expression, we have
\begin{equation*}
\begin{split}
&\frac{1}{n}\|\mathbf{f}^0-\mathbf{X}\hat{\bm{\beta}}\|^2\leq\frac{1+\alpha}{(1-\alpha)n}\|\mathbf{f}^0-\mathbf{X}\bm{\beta}\|^2+\frac{4C_1^2\sigma^2}{\alpha(1-\alpha)n}\left[sp_\mathrm{max}+s\log\left(\frac{g}{s}\right)+\log(\delta^{-1})\right] \\
&\hspace{4.5in}+\frac{4\lambda_\mathrm{max}^2}{\alpha(1-\alpha)\phi_{2s}^2},
\end{split}
\end{equation*}
holding with probability at least $1-\delta$ for $\alpha\in(0,1)$.
\end{proof}

\section{Simulations}

\subsection{Tuning parameters}
\label{app:tuning}

The range of tuning parameters for each estimator is:
\begin{itemize}
\item Group subset: a grid of $\lambda_0$ chosen adaptively using the method of Proposition~\ref{prop:lambdaseq}, where the first $\lambda_0$ sets all coefficients to zero;
\item Group subset+shrinkage: a grid of $\lambda_1$ containing logarithmically spaced points between $\lambda_1^\mathrm{max}$ and $\lambda_1^\mathrm{min}=10^{-4}\lambda_1^\mathrm{max}$, where $\lambda_1^\mathrm{max}$ is the smallest value that sets all coefficients to zero, and for each value of $\lambda_1$, a grid of $\lambda_0$ chosen as above;
\item Group lasso: a grid of $\lambda$ containing logarithmically spaced points between $\lambda^\mathrm{max}$ and $\lambda^\mathrm{min}=10^{-4}\lambda^\mathrm{max}$, where $\lambda^\mathrm{max}$ is the smallest value that sets all coefficients to zero;
\item Group SCAD: the same grid of $\lambda$ as above, and for each value of $\lambda$, a grid of the nonconvexity parameter $\gamma$ containing logarithmically spaced points between $\gamma^\mathrm{max}=100$ and $\gamma^\mathrm{min}=2+10^{-4}$; and
\item Group MCP: the same grid of $\lambda$ as above, and for each value of $\lambda$, a grid of the nonconvexity parameter $\gamma$ containing logarithmically spaced points between $\gamma^\mathrm{max}=100$ and $\gamma^\mathrm{min}=1+10^{-4}$.
\end{itemize}
Grids of 100 points are used for the primary tuning parameters ($\lambda_0$, $\lambda$) and grids of 30 points for the secondary tuning parameters ($\lambda_1$, $\gamma$).

\subsection{Comparisons with \texttt{gamsel}}
\label{app:gamsel}

The computational performance of \texttt{grpsel} is compared with \texttt{gamsel} \parencite{Chouldechova2015}, an \texttt{R} package dedicated solely to sparse semiparametric modeling via group lasso. The setup is the same as Section~\ref{sec:simulations}, except we use pseudosplines \parencite{Hastie1996} as implemented in \texttt{gamsel} for a fair comparison. The number of spline terms is again four.

Figure \ref{fig:timings-gamsel} reports the aggregated results from 30 synthetic datasets. The vertical bars are averages and the error bars are standard errors.
\begin{figure}[ht]
\centering
\input{Figures/timings-gamsel.tex}
\caption{Comparisons of \texttt{grpsel} and \texttt{gamsel}, and their estimators. Metrics are aggregated over 30 synthetic datasets generated with $\operatorname{SNR}=1$, $\rho=0.5$, and $n=1,000$. Vertical bars represent averages and error bars denote (one) standard errors.}
\label{fig:timings-gamsel}
\end{figure}
The timings for \texttt{grpsel} are virtually identical to those reported in Section~\ref{sec:simulations}, though the splines are different here. \texttt{gamsel} is slower than \texttt{grpsel} for $p=10,000$ and $p=25,000$, but the gap is relatively smaller for $p=100,000$, perhaps thanks to the screening rules described in \textcite{Chouldechova2015}. \texttt{gamsel} does not report the number of iterations required for convergence, so we cannot compare with that metric here.

\section{Data analyses}

\subsection{Macroeconomic data preprocessing}
\label{app:macro}

All series are made stationary using standard transformations given in \textcite{McCracken2016}. Some series contain missing observations and outliers, which are also treated as missing. These missing values are imputed using the \texttt{na\textunderscore kalman} function of the \texttt{R} package \texttt{imputeTS}. An observation $x_i$ is treated as an outlier if $|x_i-Q_2|/(Q_3-Q_1)>4.5$ where $Q_1$, $Q_2$, and $Q_3$ are the respective quartiles of the data.

\end{appendices}

\end{document}